\documentclass[
    aps,
    reprint,
    superscriptaddress,
    prr,
    longbibliography
]{revtex4-2}

\usepackage{bm}
\usepackage{amsmath}
\usepackage{graphicx}
\usepackage{subfigure}
\usepackage{bbold}
\usepackage{braket}
\usepackage[hidelinks]{hyperref}
\usepackage{verbatim}
\usepackage[english]{babel}
\usepackage{amsthm}
\newtheorem{theorem}{Theorem}
\newtheorem{lemma}{Lemma}
\usepackage[colorinlistoftodos]{todonotes}
\usepackage{multirow}
\usepackage{xcolor}
\usepackage{comment}

\newcommand{\be}{\begin{equation}}
\newcommand{\ee}{\end{equation}}  
\newcommand{\ls}{$\ket{\vec{n}}$ }

\definecolor{edenblue}{HTML}{378ADD}

\begin{document}

\title{Efficient Lattice Hamiltonian Encoding for the Shortest Vector Problem}
\author{Eden Schirman}
\affiliation{Physics Department, Imperial College London SW7 2BU, United Kingdom}
\affiliation{Raymond and Beverly Sackler School of Physics and Astronomy, Tel-Aviv University, Tel-Aviv 6997801, Israel}
\author{Cong Ling}
\affiliation{Electrical and Electronic Engineering Department, Imperial College London SW7 2BU, United Kingdom}
\author{Florian Mintert} \affiliation{Physics Department, Imperial College London SW7 2BU, United Kingdom}
\affiliation{Helmholtz-Zentrum Dresden-Rossendorf, Bautzner Landstraße 400, 01328 Dresden, Germany}
\date{\today}

\begin{abstract}

The advent of quantum computing necessitates the transition of worldwide cryptosystems to post-quantum cryptography (PQC),
which is founded upon the problem of finding short vectors in high-dimensional structured lattices.
It is assumed that the structure of these lattices cannot be exploited by quantum or classical algorithms attempting to find short vectors.
In this work, we focus on the structure of the lattices used in PQC protocols — nega-cyclic (and cyclic) lattices — and provide a quantum algorithmic framework that efficiently encodes the structured lattices into Hamiltonians by exploiting their underlying symmetries.
The efficient encoding substantially reduces the dimension of the corresponding Hilbert space by limiting it to a relevant subspace where short vectors are likely to be found — leading to significant savings in quantum resources (\textit{e.g.} 
qubit count and circuit depth) required to implement a quantum algorithm for finding short vectors.
We analytically prove the efficient encoding procedure and benchmark the proposed framework using the variational quantum eigensolver, demonstrating improved results with reduced quantum resources.

\end{abstract}

\maketitle

\section{Introduction}
Quantum computing poses a substantial threat to the worldwide cryptosystem \cite{RSA, JosephNaturePQCJoseph} as Shor's algorithm can efficiently solve the mathematical problem on which most cryptographic protocols are based — factoring \cite{Shor}. This threat will become tangible with the availability of fault-tolerant quantum computers \cite{EFTQC, PhysRevApplied.21.040501}. To counter it, the American National Institute of Standards and Technology (NIST) announced in 2024 a new set of cryptographic protocols believed to be secure against both quantum and classical attacks \cite{NISTIR.8547}.

The main post-quantum cryptography (PQC) protocols are based on the mathematical problem of finding the shortest non-zero vector in a high-dimensional lattice — the shortest-vector problem (SVP) \cite{IntroductionToPQCBook}. That is, given a non-orthonormal basis for a lattice, the task is to find the minimal distance between two distinct lattice points. PQC protocols are actually based on a softer version of the problem \cite{NISTIR.8547}, the approximate SVP — one where the goal is to find sufficiently short vectors whose length exceeds that of the shortest vector by at most a given polynomial factor of the lattice dimension. 

To increase the efficiency of these protocols, lattices with some special structure — structured lattices — are used \cite{PeikertDecadeOfLatticeCrypto}.
It is conjectured that the structure of these lattices does not provide any advantage to quantum (or classical) attacks \cite{IntroductionToPQCBook}, i.e., it is as hard to solve the approximate SVP for these lattices as it is for general lattices, even with a quantum computer.
Since this is still an unproven conjecture, the search for counterexamples is still open ~\cite{Cramer2021,ePrint.2025.304}.

Exploiting the structure and symmetries of a system for efficient encoding is a cornerstone technique in quantum physics. Tensor networks provide an efficient encoding of many-body quantum states by exploiting system structure \cite{VerstateTN,TNManyBodySystems}, and are used in multiple disciplines such as condensed matter physics \cite{TNCMP}, quantum chemistry \cite{ChemistryTN1},  and quantum information science \cite{TNQIS}. Quantum simulation, which deals with simulating one quantum system using another, also relies heavily on exploiting system structure and symmetries for efficient implementation \cite{QSMB, QSOpportunities, QSPractical}.

Hamiltonian encodings of lattices for the SVP have been proposed and used in previous works by some of us \cite{Joseph2TwoQuantumIsingModels, Rocky} and others \cite{JosephPRR, VQASVP, SVPAlgo2018, QALattice2024, EldarHallgren2022, MizunoWatabe2024, Engelberts2025,zheng_quantum-classical_2025,hou_efficient_2026,huang_iterative_2026}. To the best of our knowledge, this is the first work that leverages lattice structure for efficient Hamiltonian encoding, leading to a significant reduction in the quantum resources required for finding short lattice vectors.
In addition to notable savings in qubit count and circuit depth, we show that our method increases the probability of finding short vectors.

The core idea of our method is to limit the search for short vectors to a small subspace very likely to contain a state corresponding to a sufficiently short vector. The existence of such subspaces and their characterization stem from the structure of the target lattices. The likelihood of finding short vectors in these subspaces is conjectured and numerically supported, whereas the subspaces themselves are rigorously characterized.

The rest of the paper is organized as follows. The necessary preliminaries regarding lattices are described in Sec.~\ref{sec:lattices}. The symmetries of cyclic and nega-cyclic lattices that are reflected in the quantum formulation of the SVP and numerical support for our claim are described in Sec.~\ref{sec:symmetries}. The exploitation of these symmetries for efficient Hamiltonian encoding — the main result of our work — is the subject of Sec.~\ref{sec:reduced problem hamiltonians}, while Sec.~\ref{sec:constraints matrices} gives an explicit detailed description of the proposed method. Benchmarking of the method using the variational quantum eigensolver \cite{VQE} is presented in Sec.~\ref{sec: benchmarking}.

\section{Lattices}
\label{sec:lattices}

A lattice is a periodic arrangement of points in an $N$-dimensional Euclidean space. It can be defined in terms of a basis $B$ comprised of $N$ linearly independent vectors $\vec{b}_i$. While any point in the $N$-dimensional space can be expressed as a linear combination $\sum_in_i\vec{b}_i$ of the basis vectors, lattice vectors are characterized by the constraint that the expansion coefficients $n_i$ are integers. The collection of all lattice vectors is the lattice.

\subsection{Structured lattices}

Cyclic and nega-cyclic lattices are lattices with additional symmetries. Historically, cyclic lattices were firstly used in the NTRU cyptosystem \cite{hoffstein_ntru:_1998} and have received much attention since then \cite{Micciancio-cyclic,Cyclic-lattice}. Meanwhile,  cryptographic protocols in the NIST final list are actually based on nega-cyclic lattices \cite{NistThirdRoundReport,lyubashevsky_ideal_2013}. Since, however, situations with cyclic symmetries are better known in physics and since the classification for the cyclic case is rather instructive for the subsequent discussion of the nega-cyclic case, we will study both cyclic and nega-cyclic lattices. 

Their formal definition requires the cyclic permutation operator $\Pi$ such that $\Pi \vec{b}=(b_{N-1},b_0,b_1,\hdots,b_{N-2})^T$ for any vector $\vec{b}=(b_0,b_1,\hdots,b_{N-1})^T$, and the nega-cyclic permutation operator $\Gamma$, such that $\Gamma \vec{b}=(-b_{N-1},b_0,b_1,\hdots,b_{N-2})^T$.
A basis of cyclic and nega-cyclic lattices can be defined in terms of a single vector $\vec b_0$ via the relation $\vec{b}_{i+1}=\Pi \vec{b}_i$ for cyclic lattices and $\vec{b}_{i+1}=\Gamma \vec{b}_i$ for nega-cyclic lattices.

Most of the relevant properties of a lattice are encoded in its Gram matrix $G$,
comprised of all the scalar products of basis vectors;
i.e., the matrix with the elements $G_{ij}=\vec{b}_i \vec{b}_j$. 

The symmetries of the bases of cyclic and nega-cyclic lattices are also reflected in the corresponding Gram matrices, i.e., $G=\Pi G \Pi^\dagger$ for cyclic lattices and $G=\Gamma G \Gamma^\dagger$ for nega-cyclic lattices. From the unitarity of $\Pi$ and $\Gamma$ also follows the commutativity $[G,\Pi]=0$ for cyclic and $[G,\Gamma]=0$ for nega-cyclic lattices, so that the Gram matrix $G$ shares a joint set of eigenvectors with $\Pi$ or $\Gamma$ respectively.

Since the $N$-th power of the permutation reduces to the identity, the eigenvalues of $\Pi$ are given by $\exp(2\pi i j/N)$ with $j=0,\hdots, N-1$.
Analogously, the eigenvalues of $\Gamma$ are given by $\exp(\pi i (2j+1)/N)$ with $j=0,\hdots, N-1$, because of the relation $\Gamma^N=-{\bf 1}$. 

The matrix $U^{(c)}$ containing all eigenvectors of $\Pi$ is given by the Fourier transformation matrix with elements
\begin{subequations}
\label{eq:U}
\be
   U^{(c)}_{qp}=\frac{1}{\sqrt{N}}\exp\left(-i\frac{2\pi}{N}pq\right)\ ,
\ee
and the matrix $U^{(n)}$ containing all eigenvectors of $\Gamma$ has the elements
\be
    U^{(n)}_{qp}=\frac{1}{\sqrt{N}}\exp\left(-i\frac{2\pi}{2N}p\left(2q+1\right)\right)\ 
\ee
for $0\leq p, q \leq N-1$.
\end{subequations}

\subsection{Length of lattice vectors}

Since the squared Euclidean length of any lattice vector $\vec{b}=\sum_in_i\vec{b}_i$ can be expressed as
\be
\left|\vec{b}\right|^2=\sum_{ij}G_{ij}n_in_j
\ee
in terms of the Gram matrix, also a quantum mechanical Hamiltonian
\be
\label{eq:H}
\hat{H}=\sum_{ij}G_{ij}\hat{N}_i\hat{N}_j
\ee
with mutually commuting operators $\hat{N}_i$ that have integer spectra, can encode the length of all the vectors in a given lattice \cite{Joseph2TwoQuantumIsingModels}. The eigenstates of $\hat{H}$ are simultaneous eigenstates of all the operators $\hat{N}_i$ and read $\ket{\vec{n}}=\ket{n_0,n_1,\hdots,n_{N-1}}$, where $n_i$ denotes the eigenvalue of $\hat{N}_i$. 

\section{Symmetries in the problem Hamiltonian}
\label{sec:symmetries}

The quantum formulation of the approximate SVP is to find a low excited-state of the problem Hamiltonian given in Eq.~\eqref{eq:H}. This formulation is applicable for any lattice, irrespective of potential symmetries. 

Our central goal is to exploit the symmetries of cyclic and nega-cyclic lattices in order to formulate problem Hamiltonians that are beneficial for the approximate SVP. 
Their main benefit is the reduction in the number of operators $\hat{N}_q$ required for their representation. This yields a
reduction in the quantum resources needed in practical implementations on quantum devices. 

To this end, it is crucial to notice that any Hamiltonian of the form in Eq.~\eqref{eq:H} can be expressed as
\be
\hat{H}=\sum_qg_q\hat{S}_q^\dagger \hat{S}_q\ ,
\label{eq:Hd}
\ee
in terms of the eigenvalues $g_q$ of the Gram matrix and the operators
\be
\hat{S}_q=\sum_pU_{qp}\hat{N}_p\ ,
\ee
defined in terms of the unitary matrix $U$ comprised of the eigenvectors of $G$. 
Commutativity of the operators $\hat{N}_q$ implies that also the operators $\hat{S}_q$ commute. Additionally, Hermiticity of the operators $\hat{N}_p$ implies $[\hat{S}_q,\hat{S}^\dagger_p]=0$. The spectrum of $\hat{H}$ can thus be characterised completely by the simultaneous eigenstates of all the individual operators $\hat{S}_q$.

While in the case of a general lattice, the eigenvectors of $G$ depend on the details of the lattice, in the case of cyclic and nega-cyclic lattices, the eigenvectors are lattice-independent and are given by Eq. ~\eqref{eq:U}. Thus the operators $\hat{S}_q$ are completely determined by the type of symmetry and dimension of the lattice.

Since properties of a specific lattice enter the system Hamiltonian in Eq.~\eqref{eq:Hd} only via the eigenvalues $g_q$, there is %thus
the prospect of facilitating the quest for short vectors in terms of the eigenspaces of the operators $\hat{S}_q$.

By construction, any quantum state $\ket{\vec{n}}$ is an eigenstate of $\hat{S}_q$ and of $\hat{S}_q^{\dagger}\hat{S}_q$ with the corresponding eigenvalues $s^q_{\vec{n}}$ and $|s^q_{\vec{n}}|^2$. The eigenvalue of the problem Hamiltonian associated with the state $\ket{\vec{n}}$ is thus given by 

\be
\label{eq:E}
E_{\vec{n}} =\sum_qg_q|s^{q}_{\vec{n}}|^2\ .
\ee

The eigenvalues $g_q$ of the Gram matrix are non-negative due to its positive-semidefiniteness. Thus the eigenvalues $E_{\vec{n}}$ of the problem Hamiltonian (and equivalently the lengths of the lattice vectors) are given as a sum over non-negative terms. Each of these terms is associated with one of the operators $\hat{S}_q$.

One might expect that the Hamiltonian's first excited-state (which corresponds to the shortest non-zero vector of the lattice) is a ground state of most operators $\hat{S}_q$, and a low-lying excited state of only one or a few of the operators. This, however, is not necessarily the case, because the condition that a given state $\ket{\vec n}$ is a ground state of one operator $\hat S_q$ might imply that $\ket{\vec n}$ is a highly excited eigenstate to some other operator $\hat S_p$.

Still, in practice, there is a high probability of finding a low-excited state of the problem Hamiltonian in the kernel of the operator $\hat{S}_q$ which corresponds to the largest eigenvalue $g_q$ of the Gram matrix. Such an operator is referred here as the principal operator.
Crucially, the existence of kernels for the $\hat{S}_q$ operators with non-zero states $\ket{\vec{n}}$ is guaranteed because of the lattice's cyclic (nega-cyclic) structure. The sole exception is the nega-cyclic case with a power-of-two lattice dimension, whose principal operator kernel contains only the zero vector (see Sec.~\ref{sec:Explicit Characterization of Subspaces Within The Kernels}). This is based on the theory of cyclotomic polynomials ~\cite{Marcus_1977}. In order to provide a basis for the subsequent rigorous construction of low-lying eigenspaces, we will first give numerical evidence that the first excited state (corresponding to the shortest vector of the lattice) or a sufficiently low-excited state, lies within the kernel of the principal operator. 

We denote by $\gamma$ the ratio of the length of the shortest vector contained in a given subspace and the length of the lattice's shortest vector, so that $\gamma=1$ whenever the subspace contains the shortest vector. Sampling 1000 lattices per lattice dimension, Fig. ~\ref{fig:gamma=1} shows, on a logarithmic scale, the percentage of lattices with $\gamma=1$ for the principal operator kernel (blue dots) and, for comparison, for a random subspace with the same cardinality (red squares). The cardinality of a subspace denotes the number of lattice states it contains within the examined range of integer coefficient. 
The values for the principal operator kernel exceeds that for the random subspace and their expected value — the mean cardinality ratio (green triangles), by one to two orders of magnitude.

The $\gamma=1$ percentages for the principal operator kernel range from $(15.8 \pm 1.2)\%$ at $N=14$ to $(47.4\pm 1.6)\%$ at $N=9$. The statistical uncertainty on these percentages is small at this sample size, such that the variation of the $\gamma=1$ percentage across dimensions reflects the intrinsic number-theoretic structure of the kernels. The variation seems to track the mean subspace-to-lattice cardinality ratio. Overall, this indicates a high probability of finding the shortest vector within the kernel of the principal operator.

\begin{figure}[h!]
    \centering
    \includegraphics[width=1\linewidth]{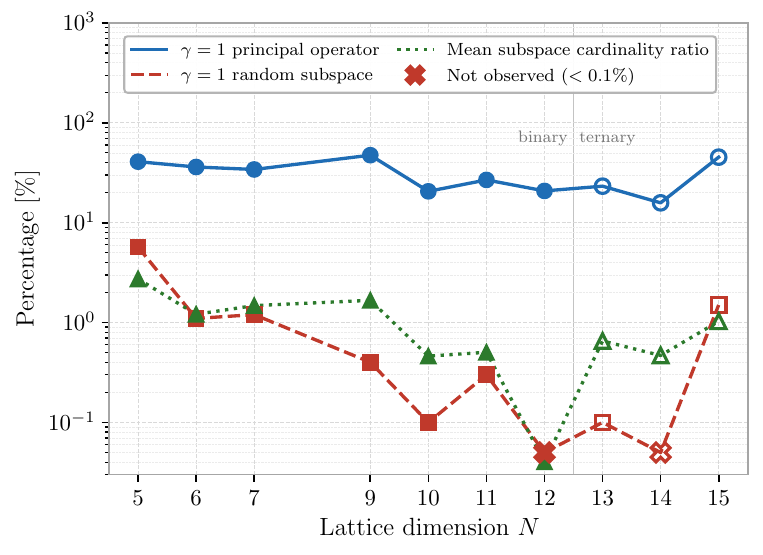}
    \caption{Percentage of sampled lattices with $\gamma=1$ (the shortest vector lies within the subspace), out of 1000 lattices per dimension, on a logarithmic scale. Values for the principal operator kernel (blue dots) are one to two order of magnitude higher than for a random subspace of the same cardinality (red squares). Green triangles show the mean subspace-to-lattice cardinality ratio, \textit{i.e.} the expected $\gamma=1$ percentage for the random subspace. 
    For $N=12$ and $N=14$ the shortest vector was not found in the random subspaces (crosses, plotted at the $<0.1\%$ detection threshold). Binary (filled markers) and ternary (hollow markers) refer to the coefficient ranges defined in the text.}
    \label{fig:gamma=1}
\end{figure}

The range of examined integer coefficients is limited by available computational resources, which impose $KN\le 24$. 
For lattice dimensions $N=5$ to $N=12$ (excluding $N=8$, a power of 2), the coefficients $n_q$ are within the range $[-2^{K-1},2^{K-1}-1]$ with $K=2$, labeled binary in Figs.~\ref{fig:gamma=1} and~\ref{fig:P90}. For $N=13,14,15$, the coefficients are restricted to $\{-1,0,1\}$ (an effective $K=\log_2(3)$), labeled ternary; this restriction is justified because the majority of the shortest vectors have coefficients within this range.

For the approximate SVP only a sufficiently short vector is required. Figure ~\ref{fig:P90} therefore shows the 90th-percentile value of $\gamma$, $P_{90}(\gamma)$, over the sampled principal operator kernels as a function of the lattice dimensions (blue dots), separated into prime (solid lines) and non-prime (dashed line) dimensions. 
As characterized in Sec. ~\ref{sec:Explicit Characterization of Subspaces Within The Kernels}, the principal operator kernel for prime $N$ is typically one dimensional. Therefore the corresponding $P_{90}(\gamma)$ values are higher than for non-prime lattice dimension, whose kernels are larger on average. Even for prime lattice dimensions, $P_{90}(\gamma)$ is on the order of $\mathcal{O}(2N)$ — linear in the lattice dimension.

For non-prime $N$, the percentile $P_{90}(\gamma)$ appears to be upper bounded by the logarithm of the inverse cardinality ratio between the smallest principal operator kernel and the full lattice. This quantity approximately equals the difference between the lattice dimension and the minimal number of kernel degrees of freedom multiplied by the range coefficient $K$, \textit{i.e.} $\sim K\times (N - min$(\#DOF)) (green triangles). The 99th-percentile values, $P_{99}(\gamma)$ (red squares), lie closer to this qualitative bound and follow the same trend.

\begin{figure}[h!]
    \centering
    \includegraphics[width=1\linewidth]{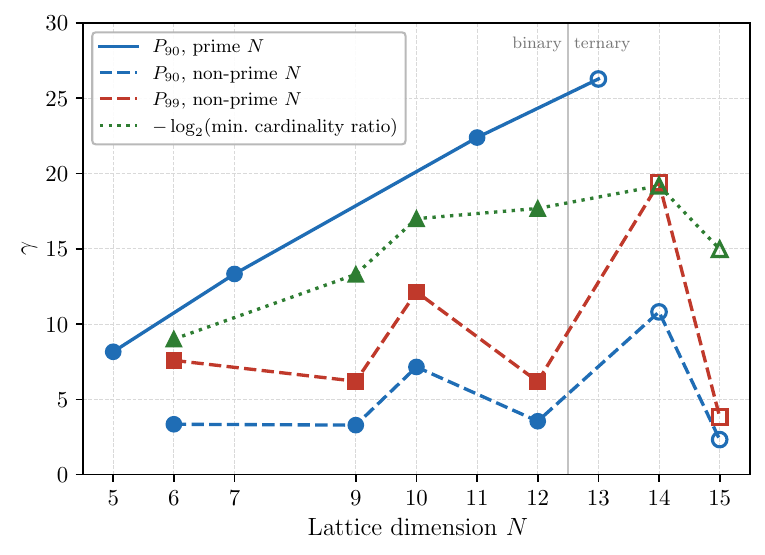}
    \caption{
    The 90th-percentile of $\gamma$, $P_{90}(\gamma)$, for the principal operator kernels: blue dots joined by a solid line for prime lattice dimensions and by a dashed line for non-prime $N$.
    Red squares show the 99th-percentile, $P_{99}(\gamma)$, for the non-prime dimensions. Both are qualitatively upper bounded by the logarithm of the inverse cardinality ratio between the smallest principal operator kernel and the full lattice (green triangles). 
    Binary (filled markers) and ternary (hollow markers) refer to the coefficient ranges defined in the text.}
    \label{fig:P90}
\end{figure}

To assess the influence of the coefficient range $K$, we repeated the simulations for $N=5$ and $N=6$ with $K=2,3,4$. For both dimensions the $\gamma=1$ percentage and $P_{90}(\gamma)$ of the principal operator kernel are independent of $K$. The only residual $K$- dependence is in $P_{99}(\gamma)$ for $N=6$, which changes from 6.78 at $K=2$ to 7.59 at both $K=3$ and $K=4$. 
The weak $K$ dependency therefore indicates that the statistics reported here are representative of the method's behavior.

The resulting prospect that low excited states (which correspond to short lattice vectors) lie within the kernel of an operator $\hat S_q$, motivates the formulation of a reduced problem Hamiltonian which represents the problem Hamiltonian with the restriction to the kernel. These reduced problem Hamiltonians
are derived in Sec.~\ref{sec:reduced problem hamiltonians}, and the kernels themselves are derived in Sec.~\ref{sec:constraints matrices} in the form of constraints matrices.

\section{Reduced Problem Hamiltonians}
\label{sec:reduced problem hamiltonians}

The likelihood of finding low excited states within the kernels of the operators $\hat{S}_q$,
motivates the definition of reduced problem Hamiltonians that allow for the search of short vectors within such kernels.
This formulation leverages the symmetries of cyclic and nega-cyclic lattices and leads to a more efficient search for low excited states of the problem Hamiltonian, i.e., for short lattice vectors.

A lattice state \ls is an eigenstate of the operator $\hat{S}_q=\sum_pU_{qp}\hat{N}_p$ with the corresponding eigenvalue $s_{\vec{n}}^q=\sum_pU_{qp}n_p$. The kernel of $\hat{S}_q$ is thus given by the span of all states \ls with integer coefficients $n_p$ that satisfy
\be
\label{eq:kernel}
\sum_{p=0}^{N-1}U_{qp}n_p=0\ .
\ee
In both the cyclic and the nega-cyclic cases, the complex coefficients $U_{qp}$ are of the form of phase factors and a constant factor $1/\sqrt{N}$ that has no influence on the solution of Eq.~\eqref{eq:kernel}. 

Since the general solution of linear algebraic equations is given by a linear combination of individual solutions,
the characterization of the kernel of an operator $\hat{S}_q$, can be represented by all integer vectors $\vec{n}=(n_0,n_1,\dots,n_{N-1})^T$ that satisfy the vector equation
\be
\label{eq:nAm}
\vec{n} = A \vec{m}\ ,
\ee
for an $\left(N \times M\right)-$dimensional constraints matrix $A$, and any $M-$dimensional integer vector $\vec{m}=(m_0,m_1,\dots,m_{M-1})^T$, with $M<N$. That is, the solutions to Eq.~\eqref{eq:kernel} are given by a constraints matrix $A$, which relates the $M$ degrees of freedom represented by any integer vector $\vec{m}$, to the vector of integers $\vec{n}$ according to Eq.~\eqref{eq:nAm}. The analytical characterizations of the constraints matrices, i.e., the explicit solutions to Eq.~\eqref{eq:kernel}, are derived in Sec.~\eqref{sec:constraints matrices}.

While the problem Hamiltonian for an $N-$dimensional lattice is given in terms of $N$ operators $\hat{N}_i$ with $\hat{H}=\sum_{ij}G_{ij}\hat{N}_i\hat{N}_j$, the characterization of the kernels of the operators $\hat{S}_q$ with associated constraints matrices $A$,  can be exploited to represent the problem Hamiltonian with a reduced number of operators $\hat{N}_i$. 

The eigenvalues of the problem Hamiltonian (equivalent to the length of lattice vectors)
that lie within the kernel of $\hat{S}_q$ with the constraints matrix $A$, are given by
\begin{subequations}
\label{eq:reduced H}
    \be
        \sum_{i,j=0}^{N-1}G_{ij}n_in_j = \sum_{i,j=0}^{M-1}F_{ij}m_im_j\ ,
    \ee 
where
\be 
\label{eq:F}
    F = A^{\dagger}GA\ .
\ee
\end{subequations}

The reduced problem Hamiltonian 
\be
\hat{H} = \sum_{i,j=0}^{M-1}F_{ij}\hat{N}_i\hat{N}_j\ ,
\label{eq:Hr}
\ee
thus represents the problem with the restriction to the kernel of $\hat{S}_q$, just like the original problem Hamiltonian formulation defined in Eq.~\eqref{eq:H}. The reduction from $N$ to $M$ operators $\hat{N}_i$, however, implies a substantial reduction in the quantum resources required in an explicit implementation, such as qubits count and circuit depth.
For example, the spectrum of the operators $\hat{N}_i$ can be represented with $K$ qubits. In that case, comparing the same implementation with and without our method, the total number of qubits saved using our approach is $K \times (N-M)$. The value of $(N-M)$ can be as large as $N-1$. It depends on the value of the lattice dimension $N$ and the index of the principal operator $q$, as rigorously characterized in the next section. A concrete example is discussed in Sec.~\ref{sec: benchmarking}.

\section{Characterization of the Constraints Matrix}
\label{sec:constraints matrices}

In this section, the constraints matrices $A$ which lead to the formulation of reduced problem Hamiltonians (Eq.~\eqref{eq:Hr}) are derived. First, the most general solution to Eq.~\eqref{eq:kernel} are characterized in Sec.~\ref{subsec:kernels} in terms of the constraints matrices. That is, the kernels of all operators $\hat{S}_q$ are characterized. 

These characterizations require a numerical method to find a solution. Therefore, in Sec.~\ref{sec:Explicit Characterization of Subspaces Within The Kernels}, another method for deriving explicit analytical solutions to Eq.~\eqref{eq:kernel} is described. These solutions yield constraints matrices which are used for explicitly representing reduced problem Hamiltonians with restrictions to subspaces within the kernels of the operators $\hat{S}_q$.

\subsection{Characterization of The Kernels}
\label{subsec:kernels}

For lattices with cyclic symmetry, the kernel of an operator $\hat{S}^{(c)}_q$ is spanned by the states $\ket{n_0,n_1,\hdots, n_{N-1}}$ with integers $n_p$ that solve the explicit form
\begin{subequations}   
\be
\label{eq:kernelcyclic} 
\vec{U}_q^{(c)}\vec{n} = 0
\ee
of Eq.~\eqref{eq:kernel} for cyclic lattices, with 
\be
\vec{U}_q^{(c)} =
\begin{pmatrix}
1 & e^{-i\frac{2\pi}{N}q} & e^{-i\frac{2\pi}{N}2q} & \hdots & e^{-i\frac{2\pi}{N}q(N-1)}
\end{pmatrix}^T\ ,
\ee
\end{subequations}
and $\vec{n}=(n_0,n_1,\dots,n_{N-1})^T$. Therefore characterizing the orthogonal space to the vector $\vec{U}_q^{(c)}$ over the integers is equal to characterizing the kernel of the operator $\hat{S}_q^{(c)}$.

The following characterization refers to the case where the greatest common divisor (gcd) of the dimension $N$ and the index $q$ is unity, i.e., $\gcd(N,q)=1$. In this case, 
the linear space spanned by the elements of $\vec{U}_q^{(c)}$ over the integers (the complementary subspace to the orthogonal space)
has dimension $\varphi(N)$~\cite{Marcus_1977}. The function $\varphi(N)$ is Euler's \emph{totient function}, which is defined as the number of integers between $1$ and $N-1$ coprime with $N$ (\textit{e.g.} for $N=6$, $\varphi(6)=2$ since only $1$ and $5$ are co-prime with $6$). The quantity $N-\varphi(N)$ is also known as the \emph{co-totient function}. 

Because the space spanned by $\vec{U}_q^{(c)}$ over the integers has dimension $\varphi(N)$, there are $\varphi(N)$ independent exponential components in the vector $\vec{U}_q^{(c)}$, and all other exponentials can be written as their integer linear combination ~\cite{Bosma,Marcus_1977}. Explicitly, the relation is given by
\be
\left(
  \begin{array}{c}
    e^{-i\frac{2\pi}{N}q(\varphi(N))} \\
    e^{-i\frac{2\pi}{N}q(\varphi(N)+1)} \\
    \vdots \\
    e^{-i\frac{2\pi}{N}q(N-1)} \\
  \end{array}
\right) = C 
\left(
  \begin{array}{c}
    1 \\
    e^{-i\frac{2\pi}{N}q} \\
    \vdots \\
    e^{-i\frac{2\pi}{N}q(\varphi(N)-1)} \\
  \end{array}
\right)\ ,
\ee
with the matrix $C$ being an integer matrix of size $(N-\varphi(N)) \times \varphi(N)$.

Therefore, Eq.~\eqref{eq:kernelcyclic}, for the case of $\gcd(N,q)=1$, can be re-written as
\begin{subequations}    
\be
\label{eq:kernel_omega_1}
  \begin{pmatrix}
    1 & e^{-i\frac{2\pi}{N}q} & \hdots & e^{-i\frac{2\pi}{N}q(\varphi(N)-1)} 
\end{pmatrix} \left( I \ |\ C^T\right)\vec{n} = 0\ ,
\ee
where $I$ is the identity matrix of size $\varphi(N) \times \varphi(N)$ and $\left( I \ |\ C^T\right)$ stands for a flat matrix concatenating $I$ and $C^T$. The phase factors $\{1,\exp\left(-i\frac{2\pi}{N}q\right), \hdots, \exp\left(-i\frac{2\pi}{N}q(\varphi(N)-1)\right)\}$ are linearly independent over the integers.
Eq.~\eqref{eq:kernel_omega_1} is thus satisfied if and only if
\be
\label{eq:kernel_omega_2}
\left( I \ |\ C^T\right)\vec{n} = \vec{0}\ .
\ee 
\end{subequations}

The solution of Eq.~\eqref{eq:kernel_omega_2} in terms of the $N-$dimensional integer vector $\vec{n}$ can be found from the {Hermite normal form} \cite[Chap. 2]{1993-cohen} of the matrix $\left( I \ |\ C^T\right)$. Given an integer matrix $X$, there is a unique matrix $Y$ in the Hermite normal form, in the form of $Y=XU$ where $U$ is a unimodular matrix, i.e., an integer matrix whose determinant is $\pm 1$. The matrix $Y$ is upper triangular, that is $Y_{ij}=0$ for $j<i+N-\varphi(N)$, and
%it also \hil{satisfies} 
$Y_{ij} < Y_{ii}$ for $j>i+N-\varphi(N)$. The Hermite normal form can be found numerically in polynomial time \cite[Chap. 2]{1993-cohen}.

{The Hermite normal form
\be
{Z} =  \left( {I}\  | \ {C}^T  \right) {U}
\ee
of the matrix $\left( I \ |\ C^T\right)$
defines a matrix $Z$ whose first $N-\varphi(N)$ columns contain only zeros.
}
We denote the first $N-\varphi(N)$ columns of the matrix ${U}$ by the submatrix $A$. 

Given the matrix $A$, the solution to Eq.~\eqref{eq:kernel_omega_2} is given by Eq.~\eqref{eq:nAm} with $\vec{n}=A\vec{m}$ ~\cite[Prop.~2.4.9]{1993-cohen}, for any $\left( N-\varphi(N)\right)$-dimensional integer vector $\vec{m}=\left(m_0,m_1,\dots,m_{N-\varphi(N)-1}\right)^T$. That is, the matrix $A$ fully characterizes the orthogonal space of the vector $\vec{U}_q^{(c)}$. Therefore, the kernel of the corresponding operator $\hat{S}^{(c)}_q$ (for the case of $\gcd(N,q)=1$) is given by all states \ls such that $\vec{n}=A\vec{m}$, for an arbitrary integer vector $\vec{m}$ and the constraints matrix $A$.

The above characterization is simply extended to cyclic operators $\hat{S}_q^{(c)}$ for which the greatest common divisor of $N$ and $q$ is greater than unity (as explained in Appendix \ref{appendix:extension}) as well as for the nega-cyclic operators $\hat{S}_q^{(n)}$, as explained in the following. For all cases however, the solution is fully determined by a constraints matrix $A$, which is used to derive reduced problem Hamiltonians according to Eq.~\eqref{eq:Hr} for the specified kernels.

For lattices with nega-cyclic symmetry, the defining equation of the kernels of the operators $\hat{S}^{(n)}_q$ 
is
\be
\label{eq:appendix negac-cyclic}
    \sum_{p=0}^{N-1}n_p\exp\left(-i\frac{\pi}{N}p\left(2q+1\right)\right)=0\  ,
\ee
(a special case of Eq.~\eqref{eq:kernel}). Here, let us assume $q=0$ for convenience, whereas the general case $q>0$ is handled in Appendix \ref{appendix:extension}. We proceed as in the case of cyclic lattices, with the difference that there are $\varphi(2N)$ independent exponential components. Again, all other exponentials can be written as their integer linear combination:
\be
\left(
  \begin{array}{c}
    e^{-i\frac{\pi}{N}(\varphi(2N))} \\
    e^{-i\frac{\pi}{N}(\varphi(2N)+1)} \\
    \vdots \\
    e^{-i\frac{\pi}{N}(N-1)} \\
  \end{array}
\right) = C 
\left(
  \begin{array}{c}
    1 \\
    e^{-i\frac{\pi}{N}} \\
    \vdots \\
    e^{-i\frac{\pi}{N}(\varphi(2N)-1)} \\
  \end{array}
\right)\ ,
\ee
with the matrix $C$ being an integer matrix of size $(N-\varphi(2N)) \times \varphi(2N)$. Then we can apply the same procedure to find the Hermite normal form, hence the kernel. Now, the kernel has dimension $N-\varphi(2N)$ rather than $N-\varphi(N)$.

\subsection{Explicit Characterization of Subspaces Within The Kernels}
\label{sec:Explicit Characterization of Subspaces Within The Kernels}

The characterization of the kernels described in Sec.~\ref{subsec:kernels} requires the use of numerical methods in order to find the constraints matrix $A$ which represents the restriction to a specific kernel. In this section however, the explicit analytical characterization of subspaces within the kernels is described. These are used for explicitly formulating reduced problem Hamiltonians restricted to these subspaces.

While the explicit characterization of these subspaces will turn out a bit cumbersome, it can be understood to some extent by the observation that an equation $\sum_pn_p \exp\left(i\phi_p\right)=0$ with the phases $\phi_p$ evenly distributed over the unit circle, has nontrivial solutions for the integer coefficients $n_p$. 

If there are $N$ evenly distributed phases, i.e., $\phi_p = \frac{2\pi}{N}p$, and $N$ is a prime number, then the general solution of the equation $\sum_pn_p\exp\left(i\frac{2\pi}{N}p\right)=0$ reads $n_0=n_1=\hdots=n_{N-1}$ with only one free integer to choose. In that case, the constraints matrix $A$ is $\left(N\times 1\right)-$dimensional and it reads $A_{i0}=1$, such that the solution is given by $\vec{n}=\left(1,1\dots 1\right)^Tm$ for an arbitrary integer $m$.

If, however, $N$ is a product of two prime numbers $Q$ and $R$ (i.e., $N=QR$), then there are at least two sets of independent solutions. The first reads $n_p=n_{p+Q}=n_{p+2Q}=\hdots=n_{p+Q(R-1)}$ with $Q$ free integer parameters $n_0,\hdots,n_{Q-1}$, and the associated constraints matrix is the $\left(N\times Q\right)-$dimensional block matrix $A_{ij}=\delta_{i\%Q,j}$ with the Kronecker delta $\delta_{i,j}$ and the modulus operation $\%$.
The second solution is $n_p=n_{p+R}=\hdots =n_{p+R(Q-1)}$ with $R$ free integer parameters $n_0,\hdots,n_{R-1}$, and the associated constraints matrix is $\left(N\times R\right)-$dimensional such that $A_{ij}=\delta_{i\%R,j}$.

An example is schematically given in Fig. \ref{fig:cyclic symmetries} for the cases of $N=3$ as a prime number and for $N=15$ as a product of prime numbers. In both plots, $N$ equally distributed phases $\frac{2\pi}{N}p$ of the unit circle are shown, with the index $p$ next to each. Solutions to the equation $\sum_pn_p\exp\left(i\frac{2\pi}{N}p\right)=0$ are schematically indicated by the colored shapes attached to each phase. Each shape represents one solution, and all coefficients having the same color within the solution need to be equal. 

For the case of $N=3$ on the left there is only one solution since 3 is a prime number, and all the coefficients must be equal ($n_0=n_1=n_2$). For the case of $N=15$ on the right, two sets of solutions indicated by circles and triangles are depicted. For each set, all integer coefficients $n_j$ %having
with the same color must be equal, \textit{e.g.} $n_0=n_5=n_{10}$ for the circles and \textit{e.g.} $n_0=n_3=n_6=n_9=n_{12}$ for the triangles.

\begin{figure}[h!]
    \centering
    \includegraphics[width=\linewidth]{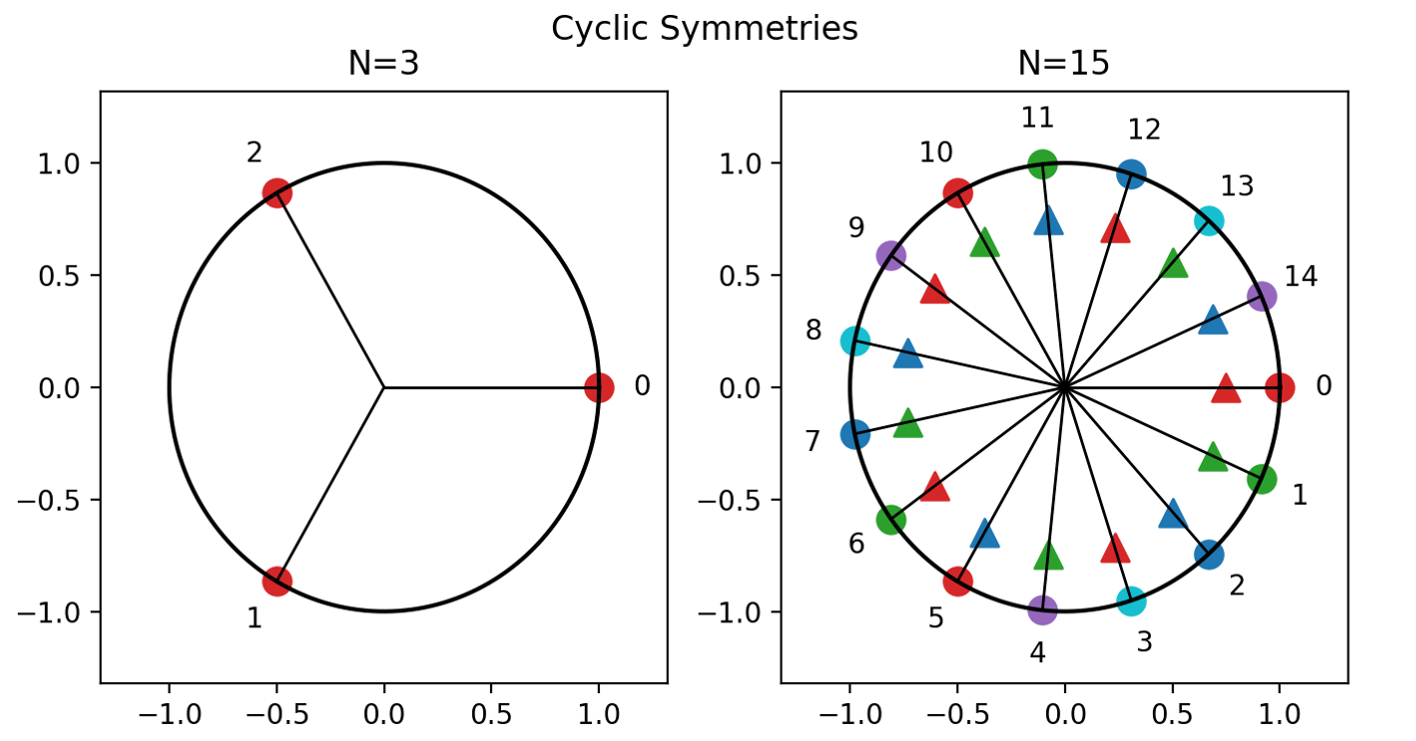}
    \caption{Schematic solutions %for
    of Eq.~\eqref{eq:kernelcyclic} for the case of $N=3$ on the left and for $N=15$ on the right. On the left, a unit circle with 3 equally distributed phases across it is depicted. All three phases are marked with a circle, meaning there is only one solution. All three circles have the same color (red) meaning their corresponding integer coefficients must be equal. On the right, next to each phase there are two shapes - circles and triangles, meaning two sets of solutions are depicted. For each solution set separately, all integer coefficients corresponding to the phases marked with shapes with the same color must be equal.}
    \label{fig:cyclic symmetries}
\end{figure}

The generalization for any dimension $N$, stems from the fact that every positive integer can be decomposed uniquely
into a product of prime factors. Solutions to Eq.~\eqref{eq:kernel} and therefore the characterization of the subspaces within the kernels, will depend on the prime factors $p$ of the lattice dimension $N$. In the simplest case, each prime factor $p$ has an associated class of states, with an associated constraints matrix. The identified subspace of the kernel is spanned by all these classes of states, and the reduced problem Hamiltonian with the associated constraints matrix thus represents the problem Hamiltonian with restriction to the subspace. 

This idea generalizes directly to more complicated cases where prime factors of different functions of $N$ are taken into account. The multiplicity of the prime coefficients for any of the cases has no influence on the structure of the kernels.

The explicit characterization of the subspaces within the kernels of each of the operators is rather technical, and the reader that is interested in the derivation that relies on number theory is referred to Appendix \ref{appendix:Explicit Characterization of Subspaces within the Kernels}. Any reader that is happy to skip the derivation can find the classification in the following, where the cyclic case is described in Sec.~\ref{sec:symmetry,cyclic} and the nega-cyclic case in Sec.~\ref{sec:symmetry, nega-cyclic}.

\subsubsection{Cyclic Case}
\label{sec:symmetry,cyclic}

The explicit form of Eq.~\eqref{eq:kernel} for lattices with cyclic symmetry is given by Eq.~\eqref{eq:kernelcyclic}. The classes of states spanning the subspaces within the kernel depend on the dimension $N$ of the lattice and on its greatest common divisor with the index $q$ of the operator $\hat{S}^{(c)}_q$, that is $K:=\gcd(N,q)$. Five cases should be distinguished:
\begin{enumerate}
    \item[(1)] $q=0$;
    \item[(2)] $K=1$ and $N$ is odd or some power of $2$;
    \item [(3)] $K=1$ and $N$ is even, but not a power of $2$;
    \item[(4)] $K>1$ and $\frac{N}{K}$ is odd or some power of $2$;
    \item[(5)] $K>1$ and $\frac{N}{K}$ is even, but not a power of $2$.
\end{enumerate}

For each case, the associated constraints matrix (matrices) $A$ is (are) explicitly characterized. The subspace is spanned by all states \ls for which $\vec{n}=A\vec{m}$, and the reduced problem Hamiltonian is given by Eq.~\eqref{eq:Hr} for each constraints matrix. The constraints matrices read: 

\begin{enumerate}
    \item[(1):] The constraints matrix $A^1$ is $\left(N \times N-1\right)-$dimensional with elements
    \be
    \label{eq:symmetry,cyclic,1}
        A^{1}_{ij}= \delta_{i,j}-\delta_{i,N-1}
    \ee
     where $\delta_{i,j}$ is the Kronecker delta (see Theorem \ref{theorem:cyc,1} in Appendix \ref{Appendix: cyclic 1,2,4}).
    
    \item[(2):] Each prime factor $p$ of $N$ has an associated $\left(N \times \frac{N}{p}\right)-$dimensional constraints matrix $A^{2,p}$ with elements:
    \be
    \label{eq:symmetry,cyclic,2}
        A^{2,p}_{ij} = \delta_{i\%\frac{N}{p},j}\ ,
    \ee
     where $\%$ is the modulo operator (see Theorem \ref{theorem:cyc,K=1} in Appendix \ref{Appendix: cyclic 1,2,4}). The associated subspace is spanned by all states \ls that obey $\vec{n}=A^{2,p}\vec{m}$ for any arbitrary integer vector $\vec{m}$ and for all matrices $A^{2,p}$ according to all prime coefficients $p$ of the lattice dimension $N$.
    
    \item[(3):] Each prime factor $p>2$ of $\frac{N}{2K}$ has an associated constraints matrix $A^{3,p}$ with dimensions $\left(N \times \left(N-\frac{N}{2Kp}(p-1)\right)\right)$. It reads
        
        \be
        \label{eq:symmetry,cyclic,3}
        \begin{split}
            A^{3,p}_{ij} =& \delta_{ij} + \sum_{k=0}^{K-1}\delta_{i+(k+\frac{1}{2})\frac{N}{K}-N,j}\\
            &+(-1)^{[j]}\left(\delta_{i\%(\frac{N}{2Kp})+k\frac{N}{K},j}-\delta_{i\%(\frac{N}{2Kp})+(k+\frac{1}{2})\frac{N}{K},j}\right)\\
            & - \sum_{l=0}^{K-2}\delta_{i+(l+1)\frac{N}{K}-N,j}\ ,
        \end{split}
        \ee
        where $[j]$ stands for the floor of the quotient of $(j+\frac{N}{2K}-N)$ and $\frac{N}{2Kp}$ (see Theorem \ref{theorem:cyclic,case 3} in Appendix \ref{Appendix: cyc, 3,5}).

    The associated subspace is spanned by all states \ls for which $\vec{n}=A\vec{m}$ for all constraints matrices $A^{2,p}$ and $A^{3,q}$ according to all prime factors $p$ of $N$ and prime factors $q$ of $\frac{N}{2K}$ , as defined in Eq.~\eqref{eq:symmetry,cyclic,2} and Eq.~\eqref{eq:symmetry,cyclic,3} respectively. Each constraints matrix is associated
    with a reduced problem Hamiltonian according to Eq.~\eqref{eq:Hr}.         
        
    \item[(4):] Each prime factor $p$ of $\frac{N}{K}$ has an associated constraints matrix $A^{4,p}$ with dimensions 
    $\left(N \times \left(N-\frac{N}{Kp}(p-1)\right)\right)$, and it reads:
    \be
    \label{eq:symmetry,cyclic,4}
        A^{4,p}_{ij} = \delta_{ij} + \delta_{i\%\frac{N}{Kp},j} + \sum_{k=1}^{K-1}\delta_{(k\frac{N}{K}+i)\%\frac{N}{Kp},j}-\delta_{k\frac{N}{K}+i-N,j}
    \ee
    (see Theorem \ref{theorem:cyclic, case 4} in Appendix \ref{Appendix: cyclic 1,2,4}).
    
    \item[(5):] The associated subspace is spanned by all states \ls for which $\vec{n}=A\vec{m}$ for all constraints matrices $A^{3,p}$ and $A^{4,q}$ according to all prime factors $p$ of $\frac{N}{2K}$ and prime factors $q$ of $\frac{N}{K}$, as defined in Eq.~\eqref{eq:symmetry,cyclic,3} and Eq.~\eqref{eq:symmetry,cyclic,4} respectively
    (see Theorem \ref{theorem:cyclic, case 4} in Appendix \ref{Appendix: cyclic 1,2,4} and Theorem \ref{theorem:cyclic,case 3} in Appendix \ref{Appendix: cyc, 3,5}).
    Each constraints matrix is associated %to
    with a reduced problem Hamiltonian according to Eq.~\eqref{eq:Hr}.
    
\end{enumerate}

\subsubsection{Nega-Cyclic Case}
\label{sec:symmetry, nega-cyclic}

In the case of nega-cyclic symmetry the explicit form of Eq.~\eqref{eq:kernel} reads
\be
\label{eq:kernel,nega-cyclic}
    \sum_{p=0}^{N-1}n_p\exp\left(-i\frac{\pi}{N}p\left(2q+1\right)\right)=0\  ,
\ee
so the kernel of the operators $\hat{S}_q^{(n)}$ is spanned by all states $\ket{n_0,n_1,\hdots,n_{N-1}}$ with integer coefficients $n_p$ that solve Eq.~\eqref{eq:kernel,nega-cyclic}.

Solutions for the general case are based on the solution for the simplest case for the operator $\hat{S}_0^{(n)}$ with the operator index $q=0$. Eq.~\eqref{eq:kernel,nega-cyclic} reduces then to
\begin{subequations}
\be
\label{eq:kernel,nega-cyclic,q=0}
	\sum_{p=0}^{N-1} n_p \exp\left(i\omega_p\right)=0\ ,
\ee
with the phases $\omega_p =-\frac{\pi}{N}p$ evenly distributed in the bottom half of the unit circle. 

If the dimension $N$ is a prime number greater than two, the solution is identified in terms of the solution to the equation $\sum_p n_p \exp\left( -i\frac{2\pi}{N}p\right)=0$ with evenly distributed phases among the whole unit circle as discussed above. This is done by adding $\pi$ to all phases $\omega_p$ with odd index $p$ in Eq.~\eqref{eq:kernel,nega-cyclic,q=0} and negating the corresponding coefficients. That is $n_p\exp\left(i\omega_p\right) =  -n_p\exp\left(i\left(\omega_p+\pi\right)\right)$ for all the integer coefficients with an odd index. Eq.  ~\eqref{eq:kernel,nega-cyclic,q=0} can then be re-written
as   
\be
    \sum_{p=0}^{N-1} y_p \exp\left(-i\frac{2\pi}{N}p\right)\ ,
\ee
where
\be
    y_p = \begin{cases}
            \quad\   n_{2p} & \text{if } p \leq \frac{N-1}{2} \\
            -n_{2p-N} & \text{if } p \geq \frac{N+1}{2}
          \end{cases}\ .
\ee
\end{subequations}
Based on the cyclic case with a prime $N$, the solution to the above equation in the integer coefficients $n_p$ has only one free parameter (\textit{e.g.} $n_0$) and it reads $n_0 = -n_1 = n_2 = -n_3 =\hdots =-n_{N-2} =n_{N-1}$.

The solution can be depicted graphically and an example for the case of $N=3$ is presented in Fig. \ref{fig:nega-cyclic symmetries}, with $3$ roots of unity evenly distributed among the bottom half of the unit circle. By negating the only odd-indexed coefficient $n_1$, its corresponding phase factor is reflected across the unit circle as indicated by the dashed black line. After the negation, the phases are evenly distributed among the unit circle thus the solution requires all integer coefficient to have the same value. The negation is depicted with the yellow cross within the circle associated to the integer coefficient $n_1$, and all circles have the same color as they are equal in absolute value.

\begin{figure}[h!]
    \centering
    \includegraphics[width=0.85\linewidth]{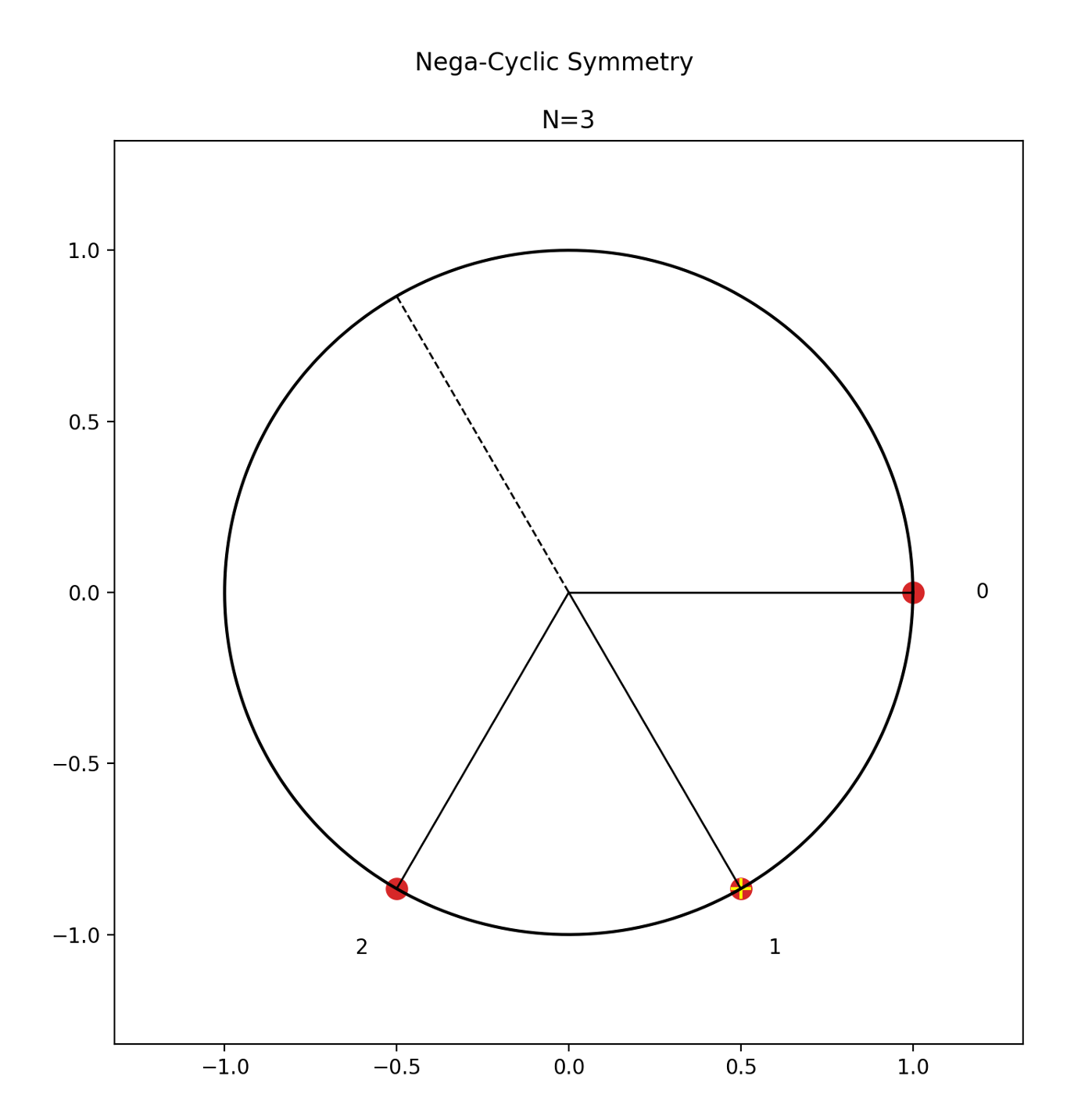}
    \caption{A schematic solution for Eq.~\eqref{eq:kernel,nega-cyclic,q=0} with $N=3$. Three phases are equally distributed in the bottom half of the unit circle. Next to each phase, a red circle is drawn, indicating that the integer coefficient associated with the phases must be equal in their absolute value. The yellow cross within the circle associated with the phase indexed $1$ indicates the negation of corresponding coefficient, i.e., $n_1 \rightarrow -n_1$, and the addition of $\pi$ radians to the corresponding phase. The flip of the phase, which is indicated by the dashed black line, results in the same distribution of phases across the entire unit circle as depicted in Fig. \ref{fig:cyclic symmetries} for the cyclic case.}
    \label{fig:nega-cyclic symmetries}
\end{figure}

As in the cyclic case, the above characterization can be extended to solutions of Eq.~\eqref{eq:kernel,nega-cyclic,q=0} for general lattice dimension $N$. Each prime factor of $N$ which is greater than two, %have
has an associated class of states, and all classes corresponding to all prime factors $p$ span the associated subspace within the kernel of $\hat{S}_0^{(n)}$. 

The characterization of subspaces within the kernel of a general nega-cyclic operator  $\hat{S}_q^{(n)}$ %depend
depends on the dimension of the lattice $N$ and on the greatest common divisor of $2N$ and $2q+1$ denoted by $L$, i.e., $L:=\gcd(2N,2q+1)$. Four cases should be distinguished:

\begin{enumerate}
    \item[(I)] $N=2^l$ for some positive integer $l$;
    \item[(II)] $L=1$ and $N$ is not some power of $2$; 
    \item[(III)] $L>1$ and $N=2^lL$ for some positive integer $l$; 
    \item[(IV)] $L>1$ and $N\ne 2^lL$ for any positive integer $l$.
\end{enumerate}

As in the cyclic case (Sec.~\ref{sec:symmetry,cyclic}), explicit constraints matrices are characterized for each case. The associated subspace is spanned by all states \ls that obey $\vec{n}=A\vec{m}$ for all the associated constraints matrices, and for an arbitrary integers vector $\vec{m}$. The reduced problem Hamiltonian associated with each constraints matrix is given by Eq.~\eqref{eq:Hr}. The constraints matrices are characterized as follows: 

\begin{enumerate}
    \item[(I):] The constraints matrix is a zero matrix such that the subspace (and the kernel) equals the zero state $\ket{0,0,\hdots,0}$ (see Theorem \ref{theorem:nega,case 1} in Appendix \ref{Appendix: nega-cyclic 1,2,3}).
    \item[(II):] Each prime factor $p>2$ of $N$ has an associated constraints matrix  $A^{II,p}$ with dimensions $\left(N \times \frac{N}{p}\right)$ and with elements
    \be
    \label{eq:symmetry,nega,2}
    A^{II,p}_{ij} = (-1)^{\lfloor i\frac{p}{N} \rfloor}\delta_{i\%\frac{N}{p},j}\ ,
    \ee
    where $\lfloor \cdot \rfloor$ is the floor operator (see Theorem \ref{theorem:nega-cyclic,case 2} in Appendix \ref{Appendix: nega-cyclic 1,2,3}).
    
    \item[(III):] The constraints matrix $A^{III}$ has dimensions $\left( N \times \left(N-\frac{N}{L}\right)\right)$ and its elements are given by
    \be 
    \label{eq:symmetry,nega,3}
    A^{III}_{ij} = \delta_{i,j} + \sum_{k=1}^{\frac{L-1}{2}} \delta_{i-(k-\frac{1}{2})\frac{N}{L},j} - \delta_{i-k\frac{N}{L},j}
    \ee 
    (see Theorem \ref{theorem:nega,case 3} in Appendix \ref{Appendix: nega-cyclic 1,2,3}).  

    \item[(IV):] Two types of constraints matrices are defined. First, for each prime factor $p$ of $\frac{2N}{L}$, a constraints matrix $A^{IV,p}$ with dimensions $\left(N \times \left(\frac{2N}{L} - \frac{2N}{Lp}\right)\right)$ is defined according to
    \begin{subequations}        
	\be
    \label{eq:symmetry,nega,4a}
    A^{IV,p}_{ij} = \delta_{i\%\frac{2N}{Lp},j} + \sum_{k=1}\delta_{i\%(\frac{2N}{Lp})+k\frac{2N}{L},j} - \delta_{i+k\frac{2N}{L},j}\ .
    \ee
 
    Secondly, for each prime factor $q>2$ of $\frac{N}{L}$, a constraints matrix $A^{V,q}$ with dimensions $\left(N \times \frac{N}{L} - \frac{N}{Lq}\right)$ is given by
	\be
    \label{eq:symmetry,nega,4b}
    \begin{split}
        A^{V,q}_{ij} = &\delta_{i,j} -\sum_{k=1}\delta_{i+k\tilde{N},j} + \sum_{l=0} \delta_{i+(l+\frac{1}{2})\tilde{N},j} +\\
				&(-1)^{\lfloor i\frac{Lq}{N} \rfloor}\left( \delta_{i\%(\frac{N}{Lq})+l\tilde{N},j} - \delta_{i\%(\frac{N}{Lq})+(l+\frac{1}{2})\tilde{N},j} \right)\ ,
	\end{split}
    \ee
    \end{subequations}
    where $\tilde{N}$ is a short hand notation for $\frac{2N}{L}$.
    
    The associated subspace within the kernel is spanned by all states \ls that obey the relation $\vec{n} = A^{IV,p}\vec{m}$ or $\vec{n} = A^{V,q}\vec{m}$ for an arbitrary vector of integers $\vec{m}$, for all prime factors $p$ of $\frac{2N}{L}$ and all prime factors $q$ of $\frac{N}{L}$ (see Theorem \ref{theorem:nega,4} in Appendix \ref{Appendix:nega,4}). Each constraints matrix has its associated reduced problem Hamiltonian according to Eq.~\eqref{eq:Hr}.
 
\end{enumerate}

\section{Benchmarking with the Variational Quantum Eigensolver}
\label{sec: benchmarking}

The reduced problem Hamiltonian $\hat{H}=\sum_{ij}F_{ij}\hat{N}_i\hat{N}_j$ defined in Eq.~\eqref{eq:Hr} represents the problem Hamiltonian $\hat{H}=\sum_{ij}G_{ij}\hat{N}_i\hat{N}_j$ defined in Eq.~\eqref{eq:H}, with the restriction to the the kernel of some operator $\hat{S}_q$ (or to a subspace within that kernel), according to the constraints matrix $A$ for which $F=A^{\dagger}GA$. The reduction in the number of operators $\hat{N}_i$ in the reduced problem Hamiltonian results in a substantial reduction in the resources needed for searching for low excited states of the Hamiltonian in any practical implementation. This method is now benchmarked with the variational quantum eigensolver algorithm \cite{VQE}.

\subsection{The Variational Quantum Eigensolver for the Problem Hamiltonian}

The variational quantum eigensolver (VQE) is a hybrid quantum-classical algorithm. Its input is some Hamiltonian $\hat{H}$ and its output is a quantum state $\ket{\Psi}$ which is highly probable to be a low excited state of the Hamiltonian $\hat{H}$ \cite{VQAReview}.

The quantum mechanical component of the algorithm is the preparation of a parameterized quantum state $\ket{\Psi}$ and a measurement scheme which allows for the evaluation of the expectation value $\bra{\Psi}\hat{H}\ket{\Psi}$ of the problem Hamiltonian $\hat{H}$ with respect to the prepared quantum state $\ket{\Psi}$. The classical part consists of a variational analysis over the parameters of the quantum state $\ket{\Psi}$ with the goal of minimizing the expectation value $\bra{\Psi}\hat{H}\ket{\Psi}$.

In any practical implementation, the states $\ket{\Psi}$ are realized with registers of qubits,
so that the problem Hamiltonians given in Eqs.~\eqref{eq:H} and~\eqref{eq:Hr} need to be defined in terms of qubit operators.
With the construction 
\be
\hat N_i=-2^{K-1}+\sum_{j=0}^{K-1}\hat Z_{ij}2^k
\ee
in terms of the Pauli-$Z$ operator for $K$ different qubits \cite{Joseph2TwoQuantumIsingModels}, the operator $\hat N_i$ has an integer spectrum ranging from $-2^{K-1}$ to $2^{K-1}-1$.
With $d$ registers of $K$ qubits each, one can thus realize a problem Hamiltonian defined in Eq.~\eqref{eq:H} for a $d$-dimensional lattice, or a reduced problem Hamiltonian as defined in Eq.~\eqref{eq:Hr} with a constraints matrix of size $\left(N\times d\right)$ for a lattice dimension $N>d$.

An Hamiltonian of the form
\begin{subequations}

\be
 \hat{H} = \sum_{ij}B_{ij}\hat{N}_i\hat{N}_j
\ee
can therefore be explicitly realized in terms of qubit operators with

\be
\begin{split}     
    \hat{H}=\sum_{ij}B_{ij} (&2^K + \sum_{pq}2^{p+q}\hat Z_{ip}\hat Z_{jq} \\ 
                                &- 2^{K-1}\sum_p(\hat Z_{ip}+\hat Z_{jp})2^p)\ .
\end{split}
\ee
\end{subequations}

In order to avoid the convergence of the VQE to the zero state $\ket{\vec{0}}$ which is the trivial ground state of Hamiltonians of the form in Eq.~\eqref{eq:H} and Eq.~\eqref{eq:Hr}, an additional penalty term, $G_{00}\ket{\vec{0}}\bra{\vec{0}}$, is added to the Hamiltonian. Thus the eigenvalue corresponding to the zero state $\ket{\vec{0}}$ is the $G_{00}$ component of the Gram matrix instead of the trivial ground state eigenvalue $0$. 

\subsection{Numerical Methods}

The applications of the VQE for the search for low-excited states of problem Hamiltonians (Eq. ~\eqref{eq:H}) and the corresponding reduced problem Hamiltonians (Eq. ~\eqref{eq:Hr}) are compared for two hundred 6-dimensional nega-cyclic lattices.
The implementations for the problem Hamiltonians contain 6 quantum registers with 3 qubits each. For the reduced problem Hamiltonians, $d$ registers of 3 qubits are used, such that $d<N$ and is determined according to the associated constraints matrix $A$ with dimensions $\left( N \times d \right)$. The analysis is made such that the spectrum of the number operators in both implementations is equal.

The explicit quantum circuits for both implementations are chosen to be standard three layer circuits followed by measurements. A scheme of one layer for a quantum circuit with 6 qubits is depicted in Fig. \ref{fig:symmetry circuit}. Each layer consists of parameterized $RY$ and $RZ$ gates applied to all qubits, followed by a circular ladder of $CNOT$ gates \cite{VQAReview}. The algorithms are emulated with the noiseless quantum simulator \textit{default.qubit} of the \textit{Pennylane} package \cite{Pennylane}.

\subsection{Examples of Reduced Qubits Count and Circuit Depth}
\label{subsec: examples}

In order to consolidate the understanding of the proposed framework, consider two specific $6-$dimensional nega-cyclic lattices with bases $B$ and $C$, that are generated from the two normalized lattice vectors
\begin{subequations}
    \be
    \vec{b}= \begin{bmatrix}
        0.010 \\ -0.45 \\ -0.50 \\ -0.67 \\ -0.36 \\ -0.18
    \end{bmatrix} 
    \ ;\quad 
    \vec{c} = \begin{bmatrix}
        -0.12  \\ -0.34\\ 0.087 \\ 0.51 \\ 0.56 \\ 0.53
    \end{bmatrix}
    \ee
\end{subequations}
respectively, represented with two significant digits. That is the basis vectors of $B$ ($C$) are given by $\{\vec{b},\Gamma\vec{b},\hdots,\Gamma^5\vec{b}\}$ ($\{\vec{c},\Gamma\vec{c},\hdots,\Gamma^5\vec{c}$\}).
The largest eigenvalue of the Gram matrix $G$ for both lattices is the $0-$th eigenvalue $g_0$. Therefore the corresponding operator $\hat{S}_0^{(n)}$ is the principal operator, and the search for low-excited states is limited to a subspace within its kernel. 

The subspace within the kernel of the operator $\hat{S}^{(n)}_0$ is characterized according to case (II) in Sec.~\ref{sec:symmetry, nega-cyclic} since the greatest common divisor of $2N$ and $2q+1$ equals unity (and six is not a power of two).
The corresponding constraints matrices are given by Eq.~\eqref{eq:symmetry,nega,2}, and because $3$ is the only prime factor greater than two of $N=6$, only one class of states spans the kernel (there is only one constraints matrix). This class of states is explicitly given by $\ket{m_0,m_1,-m_0,-m_1,m_0,m_1}$, with the free integer coefficients $m_0,m_1$.

The explicit constraints matrix $A$ which relates the free integer vector $\vec{m}=(m_0,m_1)^T$ to the $6-$dimensional integer coefficient vector $\vec{n}$ with $\vec{n}=A\vec{m}$ is explicitly given by
\be
A = 
\begin{bmatrix}
1 & 0 \\
0 & 1 \\
-1 & 0 \\
0 & -1 \\
1& 0 \\
0 & 1\\ 
\end{bmatrix}\ .
\ee
This results in the $2\times 2$ matrix $F=A^{\dagger} G A$ according to Eq.~\eqref{eq:F}. Therefore the corresponding eigenvalue of a lattice state in the kernel of $\hat{S}^{(n)}_0$ (equivalently the length of the corresponding lattice vector) is given by
\be
    \sum_{ij=0}^1 F_{ij}m_im_j + G_{00}\delta_{0,m_i}\delta_{0,m_j}\ ,
\ee
including the penalty term for the zero state.

One layer of the quantum circuit used in the VQE for the reduced problem Hamiltonian is  depicted in Fig. \ref{fig:symmetry circuit}. The full circuit contains three layers followed by a measurement of all qubits. 
 Because there are only two degrees of freedom (the constraints matrix has only two columns), the circuit contains only two quantum registers with 3 qubits each. The top three qubits  
form one register, as well as the bottom three, and upon measurement of their values, the integer coefficients $m_0$ and $m_1$ are determined respectively.

\begin{figure}[h!]
    \centering
    \includegraphics[width=\linewidth]{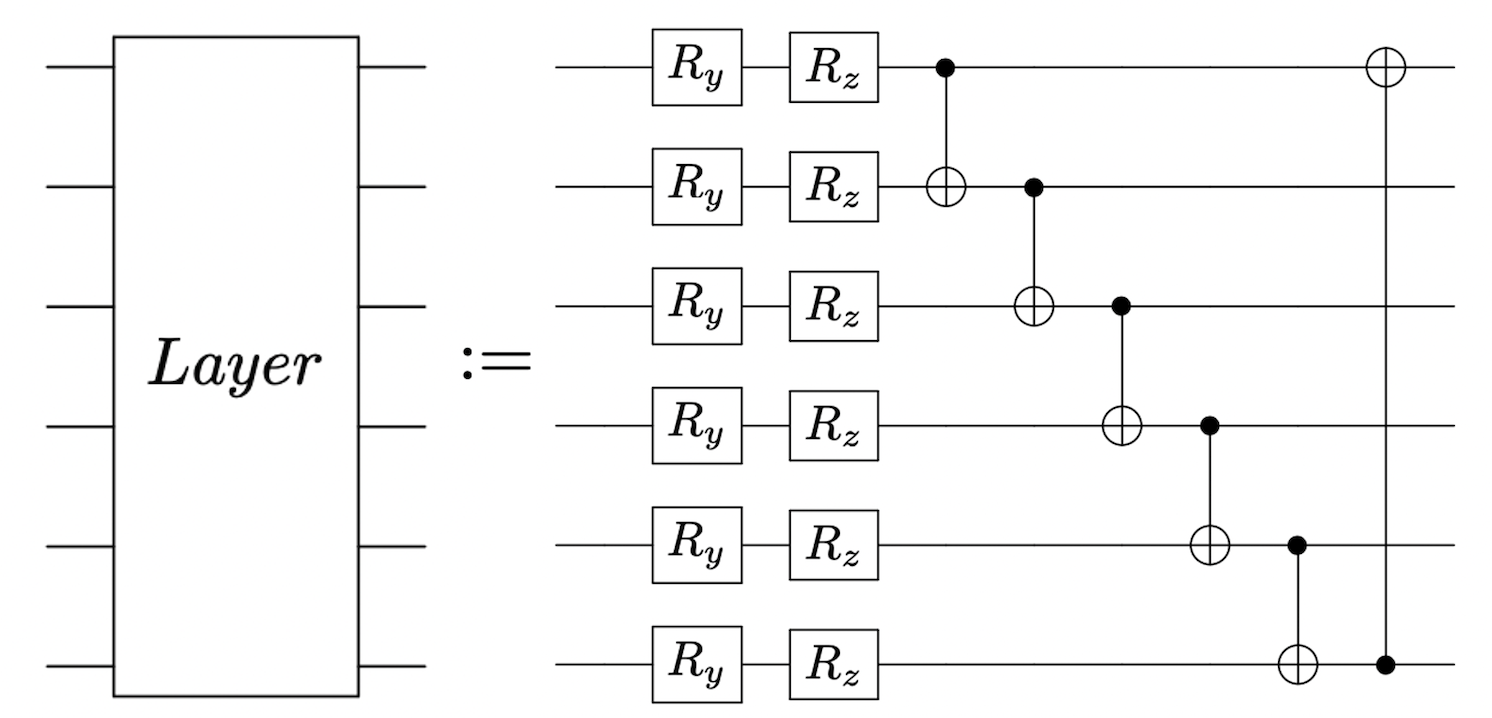}
    \caption{One layer of the variational quantum circuit for circuits with 6 qubits. The layer contains applications of the $RY$ and $RZ$ single qubit gates on all qubits, followed by circular ladder of entangling $CNOT$ gates.}
    \label{fig:symmetry circuit}
\end{figure}

The output of the VQE for the reduced problem Hamiltonian
for the lattice given by the basis $B$, is the quantum state $\ket{1,0}$ and the corresponding eigenvalue (lattice vector length) is $0.21$. The output of the VQE for the problem Hamiltonian however, which uses a quantum circuit with the same structure as in Fig ~\ref{fig:symmetry circuit} but with $18$ qubits, results in the lattice state $\ket{2,-1,-1,1,0,1}$ with the corresponding eigenvalue $0.86$. 

The ratio between the vector lengths resultant from the VQE for the reduced problem Hamiltonian and the problem Hamiltonian is denoted 
by $\lambda$. A value of $\lambda<1$ represents a lattice for which the VQE for the reduced problem Hamiltonian results in a lower excited state than the VQE for the problem Hamiltonian. This is the case for the lattice given by the basis $B$, with a value of $\lambda=0.24$.  

For the lattice given by the basis $C$, the VQE for the reduced problem Hamiltonian results in the state $\ket{0,1}$ with a corresponding vector length of $0.70$, and the VQE for the problem Hamiltonian results in the lattice state $\ket{1,-1,1,0,-1,1}$ with the corresponding vector length of $0.91$.
The value of $\lambda$ for this lattice is thus $0.77$, which is smaller than unity as well, indicating that the VQE for the reduced problem Hamiltonian results in a shorter lattice vector.

For both lattices the VQE for the reduced problem Hamiltonian results in lower excited-states than the one for the problem Hamiltonian, even though the the search space of the former is restricted to the kernel of $\hat{S}_0^{(n)}$, whilst for the latter the search space consists of all lattice states. This is possibly an outcome of a more complex optimization landscape
in the case of the full problem Hamiltonian compared to the case of the reduced problem Hamiltonian. The more complicated optimization landscape might increase the risk of the classical optimization to get stuck in a local optima rather than converging to the global minimum.

The use of only $6$ qubits for the reduced problem Hamiltonian implementation instead of $18$ for the problem Hamiltonian with a shorter circuit depth (24 instead of 60), while resulting lower excited-states, encapsulates the key benefits of our work. That is, the framework eases the search of the algorithm by restricting it to a small subspace that is very likely to contain a low excited state which corresponds to a short lattice vector, and as a consequence, 
the quantum algorithm benefits from a significant reduction in qubit count and circuit depth.

\subsection{Statistics of the Results}

A histogram which summarizes the ratio between the length of the output states of the VQE for the reduced problem Hamiltonians and for the problem Hamiltonians, i.e., the values of $\lambda$,
for all lattices examined is presented in Fig. \ref{fig:ratio histogram}. All bars to the left of the dashed black line indicate cases where the ratio is smaller than unity. This is the case for 130 out of the 200 lattices, with the median value of $\lambda$ is $0.63$.
That is, for the majority of lattices, the VQE for the reduced problem Hamiltonian results in shorter vectors than the VQE for the full problem Hamiltonian.
As argued earlier (Sec. \ref{subsec: examples}), the main reason for this is likely a better behaved optimization landscape which allows the former case to converge to a global minimum more easily.

\begin{figure}[htb]
    \centering
    \includegraphics[width=\linewidth]{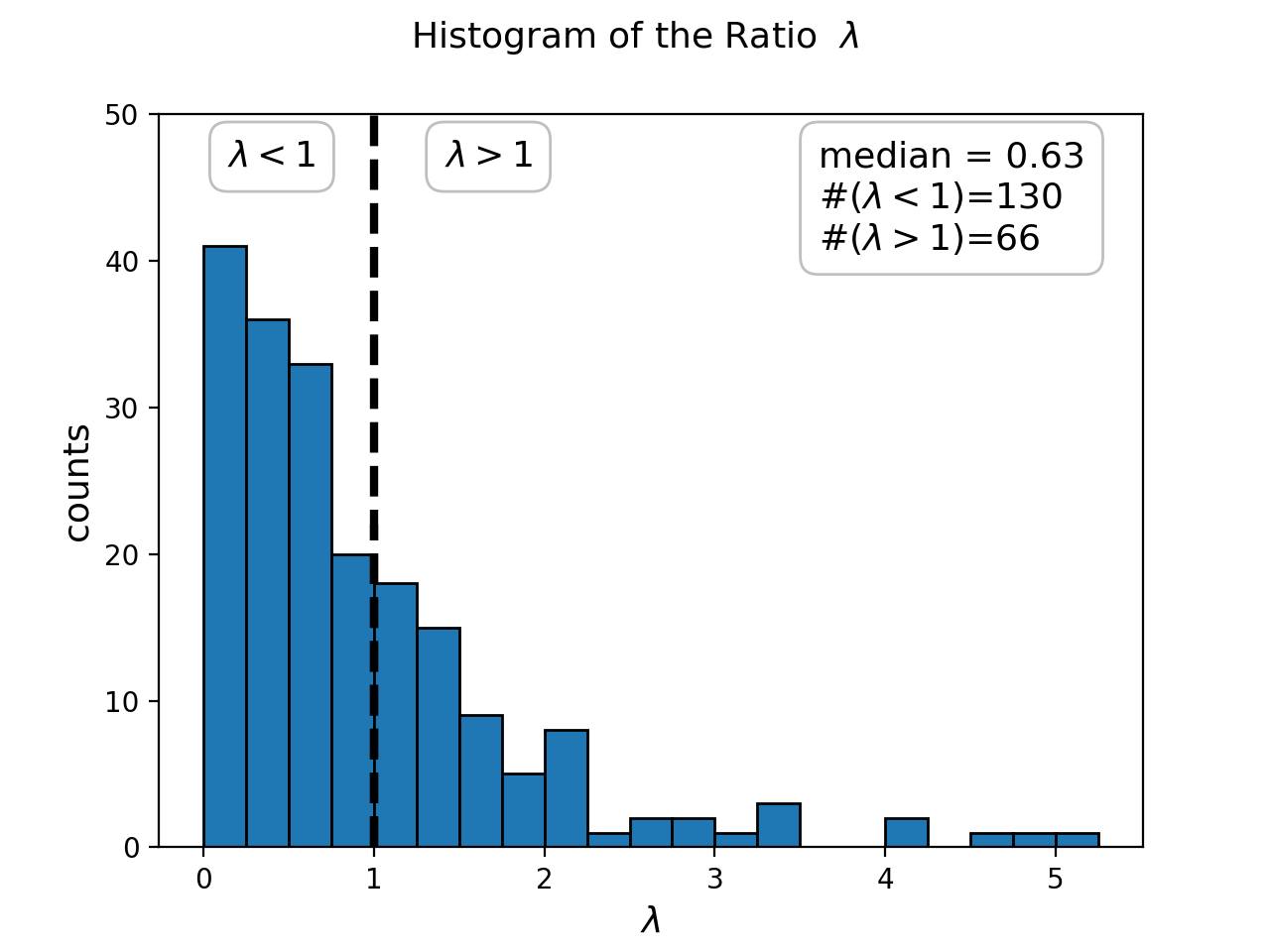}
    \caption{Histogram of the ratio between the energy values of the outputs of the VQE algorithms for the reduced problem Hamiltonians and the full problem Hamiltonians. Two hundred lattices are examined. Bars left to the dashed black line indicate lattices for which the VQE for the reduced problem Hamiltonian results in a lower excited-state than the VQE for the full problem Hamiltonian. This is the case for 130 out of the 200 lattices.}
    \label{fig:ratio histogram}
\end{figure}

The number of qubits used for the VQE implementation of the problem Hamiltonians is $18$, for all lattices, whilst the number of qubits for the implementations of the reduced problem Hamiltonians depends on the index $q$ of the operator $\hat{S}_q^{(n)}$.
For example, in the two cases examined in Sec. ~\ref{subsec: examples}, the search is restricted to a subspace within the kernel of the operator $\hat{S}_0^{(n)}$, and 6 qubits are used for the implementation of the VQE for the reduced problem Hamiltonians. The average number of qubits used for the implementations of the reduced problem Hamiltonians is $7.74$. It requires on average $2.33$ times fewer qubits than the VQE for the problem Hamiltonian. 
Overall, the application of our framework to the VQE for the search of short lattice vectors results in shorter vectors compared to the original implementation with fewer qubit counts.

\section{Outlook}

While symmetries are routinely exploited in physics to simplify computational problems, there is the common belief that symmetries of lattices are not beneficial for the shortest vector problem and its variants.
This might be true even with access to ideal quantum hardware that allows for unlimited scaling of qubit and gate counts.
As long as an increase in quantum resources is an actual bottleneck, however, the present work shows that the exploitation of symmetries can provide a gain in efficiency that ultimately makes a problem solvable.

The ability to break down an original SVP problem into several smaller SVP problems is independent of whether solutions are sought by classical or quantum mechanical means. The techniques derived in this paper can thus find applications in a broad range of classical lattice-reduction algorithms, such as the Lenstra-Lenstra-Lov\'asz algorithm~\cite{LLL} and its variants.
Because the proposed framework can be incorporated into both quantum and classical algorithms, it appears particularly promising for hybrid schemes that integrate quantum and classical components.

The availability of fault-tolerant application-scale quantum (FASQ) devices \cite{EisertPreskill2025} will open a new frontier in scientific computation. Efficient encoding of lattice-based models into Hamiltonians is therefore expected to play a central role in quantum numerical methods.
We believe our method applies to lattice-based physical models with cyclic (nega-cyclic) translational symmetry, whose interaction matrix is diagonalized by the Fourier transform  (Eq. (1)). This is the case, for example, in a classical lattice Coulomb gas of integer charges. This model is dual to the Villain form of the two-dimensional XY model — a non-trivial model exhibiting topological defects \cite{villain_theory_1975, jose_renormalization_1977}. Finding further physical models that can benefit from our method is a promising direction for future work. 

\appendix

\section{Extension of the Characterization of the Kernels}
\label{appendix:extension}

In Sec.~\ref{subsec:kernels}, the exact characterization of the kernels is given, for the case of cyclic lattices where the greatest common divisor of the the lattice dimension $N$ and the index $q$, denoted by $K$ (i.e., $K:=\text{gcd}(N,q)$) is unity. In the following, the extensions of this characterization for further cyclic and nega-cyclic cases are described. First, the extension to cyclic cases with $K>1$ is given, followed by the extension to the nega-cyclic cases with the index $q>0$.

For lattices with cyclic symmetry, the kernel of the corresponding operator $\hat{S}^{(c)}_q$ is spanned by all lattice states \ls with integer coefficients $n_p$ that satisfy
\be
\label{eq:appendix cyc}
    \sum_{p=0}^{N-1}n_p\exp\left(-i\frac{2\pi}{N}qp\right)=0\ ,
\ee
a special case of Eq.~\eqref{eq:kernel}. 

For cases with $K>1$, an effective dimension is defined as $\tilde{N}=\frac{N}{K}$ and an effective coefficient as $\tilde{q}=\frac{q}{K}$ such that $\text{gcd}(\tilde{N},\tilde{q})=1$ by definition. Eq.~\eqref{eq:appendix cyc} reduces to
\begin{subequations}
    \be
\label{eq:appendix,cyc,K>1,kernel}
\sum_{l=0}^{\tilde{N}-1}y_l\exp\left(-i\frac{2\pi}{\tilde{N}}\tilde{q}l\right)=0\ , 
\ee
where
\be
\label{eq:y_l definition}
y_l := \sum_{k=0}^{K-1}n_{l+k\tilde{N}}\ .
\ee
\end{subequations}
The solution in terms of the coefficients $y_l$ is given by the method described in Sec.~\ref{subsec:kernels}, and a direct substitution of the integers $n_p$ gives the solution that can be expressed in terms of a constraints matrix.

For lattices with nega-cyclic symmetry, since the rank is $\varphi(2N)$, we rewrite the exponential terms in \eqref{eq:appendix negac-cyclic} as 
\be
\left(
  \begin{array}{c}
    1 \\
    e^{-i\frac{2\pi}{2N}(2q+1)} \\
    \vdots \\
    e^{-i\frac{2\pi}{2N}(N-1)(2q+1)} \\
  \end{array}
\right) = D 
\left(
  \begin{array}{c}
    1 \\
    e^{-i\frac{2\pi}{2N}} \\
    \vdots \\
    e^{-i\frac{2\pi}{2N}(\varphi(2N)-1)} \\
  \end{array}
\right)\ ,
\ee
with the matrix $D$ being an integer matrix of size $N \times \varphi(2N)$. Note that $D^T$ is the counterpart of $\left( I \ |\ C^T\right)$ in \eqref{eq:kernel_omega_2}. Then we can apply the same procedure to find the Hermite normal form, hence the kernel. Now, the kernel has dimension $\geq N-\varphi(2N)$.

\section{Explicit Characterization of Subspaces within the Kernels}
\label{appendix:Explicit Characterization of Subspaces within the Kernels}

In the following, the explicit characterizations of subspaces within the kernels of the operators $\hat{S}_q$, as described in Sec.~\ref{sec:Explicit Characterization of Subspaces Within The Kernels} are proven. Due to dependencies between the proofs, first the proofs for the cyclic cases (1),(2) and (4) are given, then for the nega-cyclic cases (I), (II) and (III), and then for the rest of cyclic and nega-cyclic cases.

\subsection {Cyclic Cases (1),(2) and (4)}
\label{Appendix: cyclic 1,2,4}

The following theorems refer to lattices with cyclic symmetry. The lattice states \ls that span the kernel are all states with integer vectors $\vec{n}=A\vec{m}$ with a corresponding constraints matrix $A$. The constraints matrix represents the solutions to Eq.~\eqref{eq:appendix cyc}. The greatest common divisor of $N$ and $q$ is denoted by $K$.

\begin{theorem}
\label{theorem:cyc,1}
For $q=0$ the constraints matrix associated to the kernel of $\hat{S}^{(c)}_0$ is $\left(N \times N-1\right)-$dimensional with elements $A^{1}_{ij}= \delta_{i,j}-\delta_{i,N-1}$.
\end{theorem}
\begin{proof}
For $q=0$ Eq.~\eqref{eq:appendix cyc} reads $\sum_{p=0}^{N-1}n_p=0$ with the solution $n_{N-1}=-\sum_{p=0}^{N-2}n_p$. Thus, any $N-$dimensional vector $\vec{n}$ such that $\vec{n}=A\vec{m}$ is a solution to Eq.~\eqref{eq:appendix cyc}, for the $\left(N \times N-1\right)-$dimensional matrix $A$ with elements $A_{ij}= \delta_{i,j}-\delta_{i,N-1}$, and for any $\left(N-1\right)-$dimensional vector of integers $\vec{m}$.
\end{proof}

\begin{lemma}
\label{lemma:cyclic irreducible}
For a prime number $p$, the $p'$th cyclotomic polynomial is given by $f = 1 + \omega + \omega^2 + \dots \omega^{p-1}$, and is irreducible over the rational numbers \cite{AbstractAlgebraGarrett}.
\end{lemma}

\begin{theorem}
\label{theorem:cyc,prime N}
If $N$ is a prime number the only solution to Eq.~\eqref{eq:appendix cyc} is $n_p=n_0$ for all $p$.
\end{theorem}
\begin{proof}
For a prime number $N$, there are $N$ distinct roots of unity given by the equation
\begin{subequations}
\begin{equation}
    \omega^N=1 \iff \omega^N-1=0.
\end{equation}
The right-hand side can be factored into product of two polynomials such that
\begin{equation}
    (1+\omega+\omega^2+\dots + \omega^{N-1})(\omega-1)=0\ .
\end{equation}
\end{subequations}
The first polynomial of the product,
\begin{equation}
    f(\omega):=\sum_{p=0}^{N-1}\omega^p,
\end{equation} 
is irreducible over the rational (and integer) numbers as given by Lemma \ref{lemma:cyclic irreducible}. Thus, the roots of the equation $f=0$ are linearly independent over the field of rational numbers $\mathbb{Q}$, i.e., the unique way to express each one of the roots $\{1,e^{-i\frac{2\pi}{N}q},e^{-i\frac{2\pi}{N}q2}, \dots, e^{-i\frac{2\pi}{N}q(p-1)}\}$ as a function of all other roots, is given by the equation $f=0$. Hence, up to a multiplication with any integer, these roots of unity can sum up to zero only with all their coefficients being 1. Hence, for any prime $N$
\begin{subequations}
\begin{equation}
    \sum_{p=0}^{N-1}n_pe^{-i\frac{2\pi}{N}p}=0 \iff n_p=n_0 \quad  \forall p\in[1,N-1]\ .
\end{equation}

In addition, for a prime number $N$, the roots $e^{-i\frac{2\pi}{N}m}$ with any $m=0,1,\dots,N-1$ are primitive roots. That is, the set of taking all the powers of the root from $0$ to $N-1$ generates $N$ distinct roots of unity. Hence, for any $q\in[0,p-1]$ the set $\{e^{-i\frac{2\pi}{N}qp}\}_{p=0}^{N-1}$ is equivalent to the set $\{1,e^{-i\frac{2\pi}{N}q},e^{-i\frac{2\pi}{N}2}, \dots, e^{-i\frac{2\pi}{N}(N-1)}\}$, so the same argument as above holds. Thus for any $q\in[0,N-1]$ the following relation holds

\begin{equation}
\label{equation: cyclic prime}
    \sum_{p=0}^{N-1}n_pe^{-i\frac{2\pi}{N}qp}=0 \iff n_0=n_p \quad \forall p\in [1,N-1]\ .
\end{equation}
\end{subequations}
\end{proof}

\begin{theorem}
\label{theorem:cyc,K=1}
For $K=1$, each prime factor $p$ of the dimension $N$ has an associated class of states \ls that are within the kernel of $\hat{S}^{(c)}_q$. The class of states is defined according to Eq.~\eqref{eq:nAm}, with the $\left(N \times \frac{N}{p} \right)-$dimensional constraints matrix $A^{2,p}$ with elements $A^{2,p}_{ij} = \delta_{i\%\frac{N}{p},j}$. 
% all states of the form of $\ket{\Phi^{(c)}_{\vec{n},p_i}}$ defined in Eq.~\eqref{eq:symmetry,cyclic,2} for all prime factor $p_i$ of $N$ are within the kernel of $\hat{S}^{(c)}_q$.
\end{theorem}

\begin{proof}
For every prime coefficient $p$ of $N$, Eq.~\eqref{eq:appendix cyc} can be explicitly written as:

\begin{subequations}
\label{eq:appendix,cyc,case 2}
\be
    \sum_{l=0}^{\frac{N}{p}-1}e^{-i\frac{2\pi }{N}ql}f_l(e^{-i\frac{2\pi}{p}q}) = 0\ ,
\ee
with:
\be
\label{appendix:f_l}
    f_l(\omega) := \sum_{m=0}^{p-1}n_{l+m\frac{N}{p}}\omega^{m}\ .
\ee 
\end{subequations}

The polynomial $f_l(\omega)$ follows Theorem \ref{theorem:cyc,prime N} for prime dimensions. 
Thus, Eq.~\eqref{eq:appendix,cyc,case 2} can be solved if all the integer coefficients $n_{l+m\frac{N}{p}}$ for each polynomial $f_l$ are equal. 

This results in $\frac{N}{p}$ degrees of freedom, where all states of the form $\ket{n_0,\hdots,n_{\frac{N}{p}-1},n_0,\hdots,n_{\frac{N}{p}-1},n_0,\hdots}$ are within the kernel of $\hat{S}^{(c)}_q$. Therefore, all lattice states \ls with $\vec{n}=A^{2,p}\vec{m}$ for an $\frac{N}{p}$-dimensional vector $\vec{m}$ and an $\left(N \times \frac{N}{p} \right)-$dimensional matrix $A^{2,p}$ with elements $A^{2,p}_{ij} = \delta_{i\%\frac{N}{p},j}$ are within the kernel of $\hat{S}^{(c)}_q$ for $K=1$.
\end{proof}

\begin{theorem}
\label{theorem:cyclic, case 4}
For $K>1$, for any prime factor $p$ of $\frac{N}{K}$ there is a class of states \ls that are within the kernel of $\hat{S}^{(c)}_q$. This class is defined according to Eq.~\eqref{eq:nAm} with the constraints matrix $A^{4,p}$ with dimensions $\left(N \times \left(N-\frac{N}{Kp}(p-1)\right)\right)$ and with elements defined in Eq.~\eqref{eq:symmetry,cyclic,4}.
\end{theorem}

\begin{proof}
For $K>1$, an effective dimension is defined as $\tilde{N}=\frac{N}{K}$ and an effective coefficient as $\tilde{q}=\frac{q}{K}$ such that $\text{gcd}(\tilde{N},\tilde{q})=1$. Eq.~\eqref{eq:appendix cyc} reduces to
\begin{subequations}
    \be
\label{eq:appendix,cyc,K>1}
\sum_{l=0}^{\tilde{N}-1}y_l\exp\left(-i\frac{2\pi}{\tilde{N}}\tilde{q}l\right)=0\ , 
\ee
where
\be
\label{eq:y_l definition}
y_l := \sum_{k=0}^{K-1}n_{l+k\tilde{N}}\ .
\ee
\end{subequations}

The solution in terms of the coefficients $y_l$ is given by Theorem \ref{theorem:cyc,K=1} such that for each prime factor $p$ of $\tilde{N}$, the solution has $\frac{\tilde{N}}{p}$ free integer coefficients and it reads
\begin{subequations}    
\be
\label{eq:y_l solution}
    y_l = y_{[l]}\ ,
\ee
with $[l]$ denoting $l$ mod $\frac{\tilde{N}}{p}$. 

Substituting Eq.~\eqref{eq:y_l definition} into Eq.~\eqref{eq:y_l solution} leads to the solution:
\be
    n_j = n_{[j]} + \sum_{k=1}^{K-1}n_{[k\frac{N}{K}+j]}-n_{k\frac{N}{K}+j-N}\ .
\ee
\end{subequations}
For each prime factor $p$ of $\tilde{N}$ there are $\tilde{N}-\frac{\tilde{N}}{p}$ constraints and  $N-\tilde{N}-\frac{\tilde{N}}{p}$ free modes. These are encapsulated by the constraints matrix $A^{4,p}$ with dimensions $\left(N \times \left(N-\frac{N}{Kp}(p-1)\right)\right)$, and with elements defined in Eq.~\eqref{eq:symmetry,cyclic,4}.
\end{proof}

\subsection{Nega-cyclic Cases (I),(II) and (III)}
\label{Appendix: nega-cyclic 1,2,3}

The following theorems refer to lattices with nega-cyclic symmetry. The lattice states that span the kernel are all states with integer coefficients that solve Eq.~\eqref{eq:appendix negac-cyclic}. The greatest common divisor of $2N$ and $2q+1$ is denoted by $L$.

\begin{lemma}
\label{lemma: characteristic of cyclo}
The $n-$th primitive root of unity $\zeta_n$, is a root of a given polynomial $f(\zeta_n)=0$ \textit{iff} the $n-$th cyclotomic polynomial $\Phi_n(\omega)$, divides the polynomial $f$, that is $\Phi_n(\omega)|f(\omega$) \cite{AbstractAlgebraGarrett}.
\end{lemma}

\begin{theorem}
\label{theorem:nega,case 1}
    If $N=2^m$ for some non-negative integer $m$, the kernel of $\hat{S}^{(n)}_q$ is spanned by the zero state $\ket{\vec{0}}$.
\end{theorem}

\begin{proof}
For the case of $q=0$, Eq.~\eqref{eq:appendix negac-cyclic} reduces to $f(\omega)=0$ with $\omega=\exp\left(-i\frac{2\pi}{2N}\right)$ and 
\be
    f(\omega):=\sum_{l=0}^{N-1}n_l\omega^l\ .
\ee
The $2N-$th cyclotomic polynomial is given by $\Phi_{2N}(\omega)=\omega^{N}+1$. The degree of the cyclotomic polynomial is $N$ whilst the degree of the the polynomial $f(\omega)$ is $N-1$. Hence the cyclotomic polynomial cannot divide the $f(\omega)$, and according to Lemma \ref{lemma: characteristic of cyclo} there are no primitive roots for the dimension $2N$ that are roots of the polynomial $f$, specifically $\omega=\exp\left(-i\frac{2\pi}{2N}\right)$. That is, there is no solution to the equation $f(\exp\left(-i\frac{2\pi}{2N}\right))$ unless all the coefficients $n_l$ are zero, i.e., the trivial solution. 

For other cases $q>0$, the same argument applies, since $\omega^{2q+1}$ ($1\leq q \leq N-1$) are exactly the other roots of the cyclotomic polynomial $\Phi_{2N}(\omega)$.
\end{proof}

\begin{theorem}
\label{theorem:negacylic, prime number}
If $N$ is a prime number, the only non-trivial solution for Eq.~\eqref{eq:appendix negac-cyclic} is $n_0=-n_1=n_2=\hdots=-n_{N-2}=n_{N-1}$.
\end{theorem}

\begin{proof}
By adding $\pi$ radians to the phase of each odd-indexed integer coefficient $n_l$ (i.e., $l$ is odd), Eq.~\eqref{eq:appendix negac-cyclic} can be re-written as:
\begin{subequations}    
\label{eq:appendix,nega-cyclic,prime}
\be
    \sum_{l=0}^{N-1}y_l\exp\left(-i\frac{2\pi}{N}(2q+1)l\right)\ ,
\ee
with
\be 
y_l = \begin{cases}
        \quad\   n_{2l} & \text{if } l \leq \frac{N-1}{2} \\
        -n_{2l-N} & \text{if } l \geq \frac{N+1}{2}
      \end{cases}\ .
\ee
\end{subequations}
The solution is given by Theorem \ref{theorem:cyc,prime N} in terms of the integer coefficients $y_l$. Substituting the original coefficients leads to the only solution $n_0=-n_1=n_2=\hdots=-n_{N-2}=n_{N-1}$.
\end{proof}

\begin{theorem}
\label{theorem:nega-cyclic,case 2}
    If $L=1$ and $N$ is not some power of two, for each prime factor $p$ of the dimension $N$, all states \ls defined according to Eq.~\eqref{eq:nAm} with the constraints matrix $A^{II,p}$ are within the kernel of $\hat{S}^{(n)}_q$. The constraints matrix $A^{II,p}$ is $\left(N \times \frac{N}{p}\right)$-dimensional and its elements are defined according to Eq.~\eqref{eq:symmetry,nega,2}.
\end{theorem}

\begin{proof}
    Following the proof of Theorem \ref{theorem:cyc,K=1}, for every prime factor $p>2$ of $N$, Eq.~\eqref{eq:appendix negac-cyclic} can be decomposed to $\frac{N}{p}$ polynomials (as in Eq.~\eqref{eq:appendix,cyc,case 2}):
    \begin{subequations}
        \be
        \sum_{l=0}^{\frac{N}{p}-1}e^{-i\frac{2\pi}{2N}(2q+1)l}f_l(e^{-i\frac{2\pi}{2p}(2q+1)}) = 0\ ,
        \ee
    \end{subequations}
    where the polynomial $f_l(\omega)$ is defined in Eq.~\eqref{appendix:f_l}. The integer coefficients of each polynomial must have the form defined in Theorem \ref{theorem:negacylic, prime number} in order for the equation to be solved. 

    Thus, for each prime factor $p$ of $N$, all states of the form $\ket{n_0,\dots,n_{\frac{N}{p}-1},-n_0,\dots,-n_{\frac{N}{p}-1},n_0,\dots}$ span a subspace within the kernel of $\hat{S}^{(n)}_{q}$. These are represented by the constraints matrix $A^{II,p}$ with elements defined in Eq.~\eqref{eq:symmetry,nega,2}.  
\end{proof}

\begin{theorem}
\label{theorem:nega,case 3}
    If $L>1$ and $N=2^mL$ for some positive integer $m$, 
    all states \ls defined according to Eq.~\eqref{eq:nAm} with the constraints matrix $A^{III}$ are within the kernel of $\hat{S}^{(n)}_q$. The constraints matrix $A^{III}$ is $\left(N \times \left( N-\frac{N}{L}\right)\right)-$dimensional and its elements are defined according to Eq.~\eqref{eq:symmetry,nega,3}.
\end{theorem}

\begin{proof}
By defining an effective dimension $\tilde{N}:=\frac{2N}{L}$ and an effective index $\tilde{q}:=\frac{2q+1}{L}$, Eq.~\eqref{eq:appendix negac-cyclic} reduces to 
\begin{subequations}        
    \be
    \sum_{l=0}^{\tilde{N}-1}y_l\exp\left(-i\frac{2\pi}{\tilde{N}}\tilde{q}l\right)=0\ ,
    \ee
    with
    \be
    y_l := \sum_{k:[k]=l}n_k\ ,
    \ee
\end{subequations}
where $[k]$ is a short-hand notation for $k$ mod $\tilde{N}$. Because the effective dimension is some power of two and $\tilde{K}:=\text{gcd}(\tilde{N},\tilde{q})=1$, the solution is given by Theorem \ref{theorem:cyc,K=1} in terms of the integers $y_l$. 

The substitution by the original integer coefficients $n_l$, alongside the fact that $2$ is the only prime factor of $\tilde{N}$ leads to $\frac{N}{L}$ constraints that are encapsulated by the $\left(N \times \left(N-\frac{N}{L}\right)\right)-$dimensional constraints matrix $A^{III}$ with elements defined in Eq.~\eqref{eq:symmetry,nega,3}.
\end{proof}

\subsection{Cyclic Cases (3) and (5)}
\label{Appendix: cyc, 3,5}

\begin{theorem}
\label{theorem:cyclic,case 3}
For any $K$ such that $\frac{N}{K}$ is even and is not some power of two, each prime factor $p$ of $\frac{N}{2K}$ has an associated class of states \ls that are within the kernel of $\hat{S}^{(c)}_q$. The class of states is defined by the $\left(N \times \left(N-\frac{N}{2Kp}(p-1)\right)\right)$-dimensional constraints matrix $A^{3,p}$ with elements defined in Eq.~\eqref{eq:symmetry,cyclic,3}.   
\end{theorem}

\begin{proof}
Eq.~\eqref{eq:appendix cyc} can be written as
\begin{subequations}
    \be 
    \label{eq:appendix,case 3, part 1}
    \sum_{l=0}^{\tilde{N}-1}y_l\exp\left(-i\frac{2\pi}{\tilde{N}}\tilde{q}l\right)= 0 \ ,
    \ee 
with $\tilde{N}=\frac{N}{K}$, $\tilde{q}=\frac{q}{K}$ and $y_l = \sum_k n_{l+k\tilde{N}}$. Because $\tilde{N}$ is even and the greatest common divisor $\text{gcd}(\tilde{N},\tilde{q})$ equals unity by definition, $\tilde{q}$ is odd. Thus the equation can be re-written as
    \be 
    \sum_{l=0}^{\frac{\tilde{N}}{2}-1}z_l\exp\left(-i\frac{2\pi}{2\frac{\tilde{N}}{2}}(2q'+1)l\right)=0\ ,
    \ee 
    \end{subequations}
with $z_l=y_l-y_{l+\tilde{N}/2}$ and $\tilde{q}=2q'+1$. The last equation can be recognized as the defining equation for the nega-cyclic symmetry, i.e., Eq.~\eqref{eq:appendix negac-cyclic}, with $L=1$ and $N$ which is not some power of $2$ (i.e., case (II) for nega-cyclic lattices).

Hence, according to Theorem \ref{theorem:nega-cyclic,case 2} each prime factor $p$ of $\frac{\tilde{N}}{2}$ has an associated class of states that lies within the kernel of $\hat{S}^{(c)}_{q}$. Expressing the integer coefficients $z_l$ in terms of the original integer coefficients $n_l$ leads to the symmetry classes being defined in terms of the constraints matrix $A^{3,p}$ for each prime factor $p$ of $\frac{N}{2K}$ with elements defined in Eq.~\eqref{eq:symmetry,cyclic,3}.
\end{proof}

\subsection{Nega-cyclic Case (IV)}
\label{Appendix:nega,4}

\begin{theorem}
\label{theorem:nega,4}
    If $L>1$ and $N\ne2^mL$ for any positive integer $m$, a subspace within the kernel of $\hat{S}^{(n)}_q$ is spanned by all states \ls that obey the relation $\vec{n} = A^{IV,p}\vec{m}$ or $\vec{n} = A^{V,q}\vec{m}$ for an arbitrary vector of integers $\vec{m}$, for all prime factors $p$ of $\frac{2N}{L}$ and all prime factors $q$ of $\frac{N}{L}$.    
\end{theorem}

\begin{proof}
By defining the effective dimension as $\tilde{N}=\frac{2N}{L}$ and the effective index as $\tilde{q}=\frac{2q+1}{L}$, Eq.~\eqref{eq:appendix negac-cyclic} reduces to 
\begin{subequations}    
\be 
    \sum_{l=0}^{\tilde{N}-1}y_l\exp\left(-i\frac{2\pi}{\tilde{N}}\tilde{q}l\right)=0
\ee 
with 
\be
    y_l = \sum_{[k]=l}n_k\ ,
\ee
\end{subequations}
where $[k]$ is a short hand notation for $k$ mod $\tilde{N}$. Because the greater common divisor of $\tilde{N}$ and $\tilde{q}$ is unity by definition (i.e., $\tilde{K}:=\text{gcd}(\tilde{N},\tilde{q})=1$) and $\tilde{N}$ is even but not some power of two, the solutions to the above equation are given by cyclic case (3) in terms of the integer coefficients $y_l$. Substituting the original coefficients $n_l$ and expressing the solution in terms of constraints matrices, results in the two types of constraints matrices.

First, for every prime factor $p$ of $\frac{2N}{L}$, the $\left(N \times \left(\frac{2N}{L} - \frac{2N}{Lp}\right)\right)$-dimensional constraints matrix $A^{IV,p}$ is given explicitly by Eq.~\eqref{eq:symmetry,nega,4a}. Secondly, for each prime factor $q>2$ of $\frac{N}{L}$ the constraints matrix $A^{V,q}$ with dimensions $\left(N \times \left(\frac{N}{L} - \frac{N}{Lq}\right)\right)$ is given by Eq.~\eqref{eq:symmetry,nega,4b}. A subspace within the kernel of $\hat{S}^{(n)}_q$ is spanned by all states \ls that obey the relation $\vec{n} = A^{IV,p}\vec{m}$ or $\vec{n} = A^{V,q}\vec{m}$ for an arbitrary vector of integers $\vec{m}$.
\end{proof}

\bibliography{bib}
\end{document}